\documentclass[times,sort&compress,3p]{elsarticle}
\journal{Journal of Multivariate Analysis}
\usepackage[labelfont=bf,labelsep=space]{caption}

\usepackage{hyperref}
\usepackage{amsfonts}
\usepackage{amsmath}
\usepackage{amssymb}
\usepackage{amsthm}
\usepackage{algorithm}
\usepackage{algorithmic}
\usepackage{graphicx}
\usepackage{color}
\usepackage[labelfont=bf]{subcaption}
\usepackage{mathrsfs}
\usepackage[utf8]{inputenc}

\renewcommand{\vec}[1]{\ensuremath \mathbf{\boldsymbol{#1}}}
\newcommand{\N}{\ensuremath{\mathbb{N}}}

\renewcommand{\S}{\ensuremath{\mathbb{S}}}

\newcommand{\Z}{\ensuremath{\mathbb{Z}}}

\newcommand{\R}{\ensuremath{\mathbb{R}}}

\newcommand{\E}{\ensuremath{\mathbb{E}}}
\newcommand{\SO}{\ensuremath{\mathrm{SO}(3)}}

\newcommand{\abs}[1]{\ensuremath{\left\vert#1\right\vert}}
\newcommand{\dx}{\mathrm{d}}

\DeclareMathOperator*{\argmax}{argmax}
\DeclareMathOperator*{\tr}{tr}

\renewcommand{\d}{\, \mathrm{d}}

\newcommand{\norm}[1]{\left\lVert \smash{#1} \right\rVert}
\newcommand{\Norm}[1]{\left\lVert #1 \right\rVert}

\newtheorem{theorem}{Theorem}[section]
\newtheorem{lemma}[theorem]{Lemma}
\newtheorem{remark}[theorem]{Remark}
\newtheorem{definition}[theorem]{Definition}

\newtheorem{corollary}[theorem]{Corollary}

\newcommand{\dotprod}[2]{ \left< #1,#2 \right>}
\newcommand{\changed}[1]{{#1}}

\begin{document}
\begin{frontmatter}

\title{\changed{Locally} Isometric Embeddings of Quotients of the Rotation Group Modulo Finite Symmetries}

\author[A1]{Ralf Hielscher\corref{mycorrespondingauthor}}
\author[A1]{Laura Lippert}

\address[A1]{Fakultät für Mathematik, Technische Universität Chemnitz}

\cortext[mycorrespondingauthor]{Corresponding author.\\ Email addresses: \url{ralf.hielscher@mathematik.tu-chemnitz.de}, \url{laura.lippert@mathematik.tu-chemnitz.de}}

\begin{abstract}
  The analysis of manifold\changed{-}valued data using embedding based methods is linked
  to the problem of finding suitable embeddings. In this paper we are
  interested in embeddings of quotient manifolds $\SO/\mathcal S$ of the
  rotation group modulo finite symmetry groups. Data on such quotient
  manifolds naturally occur in crystallography, material science and
  biochemistry. We provide a generic framework for the construction of such
  embeddings which generalizes the embeddings constructed in \cite{ArJuSc18}. The
  central advantage of our larger class of embeddings is that it \changed{includes
  locally} isometric embeddings for all crystallographic symmetry groups.
\end{abstract}


\end{frontmatter}


\section{Introduction}

In the analysis of manifold\changed{-}valued data there are two different approaches -
intrinsic and extrinsic. Intrinsic methods solely rely on intrinsic
properties of the manifold, e.g. the Riemanian curvature tensor, the
exponential map or the Levi-Cevita connection. Those methods often work
locally like moving least squares \cite{GrSpYu17}, multiscale methods
\cite{RaDrStDoSc05} or subdivision schemes \cite{XieYu}. Other intrinsic
approaches make use of function systems that are adapted to the geometry of
the manifold, e.g. diffusion maps \cite{CoLa06} or the eigenfunctions of the
manifold Laplacian \cite{jupp2005,he90, kim98, Pe05, Hi13a}.

On the other hand, extrinsic methods rely on an embedding of the manifold into
some higher dimensional vector space \cite{ArJuSc18,RoTa14,CrSt13}. The
advantage of \changed{embedding-based methods, compared to intrinsic methods}, is that they often are straight forward
generalizations of the corresponding linear methods. The central challenges
for applying an embedding\changed{-}based method to a specific manifold $\mathcal M$ are
\begin{enumerate}
\item Find a suitable embedding
  $\mathcal E \colon \mathcal M \to \mathbb R^{d}$ of the manifold
  $\mathcal M$ that approximately preserves distances and has moderate dimension.
\item Find an efficient algorithm for the projection
  $P_{\mathcal M} \colon U \to \mathcal M$ from
  some neighborhood $U \supset \mathcal E(\mathcal M)$ back to the manifold.
\end{enumerate}

In our paper we are concerned with the specific case when the manifold
$\mathcal M$ is the quotient
$\SO/\mathcal S = \{[\vec R]_{\mathcal S} \colon \vec R \in \SO\} $ of the
rotational group $\SO$ with respect to some finite symmetry group
\changed{$\mathcal S < \SO$}. \changed{Here the cosets in the quotient space
  are defined by
  $[\vec R]_\mathcal S:=\{\vec R\vec O \mid \vec O\in \mathcal S \}$. As a
  finite subgroup of $\SO$ the symmetry group $\mathcal S$ is isomorphic to
  one of the following: the cyclic groups $C_k$ for
  $k\in \left\{1,2,\ldots\right\}$, the dihedral groups $D_k$ for
  $k\in \left\{2,3,\ldots\right\}$, the tetrahedral group $T$, the octahedral
  group $O$ and the icosahedral group $Y$. Since the group $\SO$ is simple,
  the quotient $\SO/\mathcal S$ is not a group for all $\mathcal S \ne C_{1}$
  but forms a homogeneous space with canonical left action of the Lie group
  $\SO$.}

\changed{To give the reader an idea about the quotient $\SO/\mathcal S$ we
  consider the representation of a rotation
  $\vec R = \vec R_{z}(\alpha) \vec R_{y}(\beta) \vec R_{z}(\gamma)$ as the
  composition of rotations about the axes $z$, $y$, $z$ and Euler angles
  $\alpha, \gamma \in [0,2\pi]$, $\beta \in [0,\pi]$. Let us furthermore
  assume that the subgroup $C_{k}$ is represented by the rotations
  $\vec R_{z}(\ell\,\frac{2\pi}{k})$, $\ell \in \Z$ about the z-axis.  Then
  $C_{k}$ enforces a periodicity of $2\pi/k$ on the last Euler angle $\gamma$
  and the cosets in $\SO/C_{k}$ are of the form
  \begin{equation*}
    [\vec R_{z}(\alpha) \vec R_{y}(\beta) \vec R_{z}(\gamma)]_{C_{k}}
    = \left\{\vec R_{z}(\alpha) \vec R_{y}(\beta) \vec R_{z}(\gamma + \tfrac{2\ell \pi}{k})
  \mid \ell = 0 \ldots k-1 \right\}.
  \end{equation*}
  Nice geometrical visualizations of these coset spaces can be found in
  \cite{krakow17}.}

\changed{The analysis of data that are cosets
  $[\vec R]_{\mathcal S} \in \SO / \mathcal S$ in the homogeneous space
  $\SO / \mathcal S$ is of central importance in various scientific areas.}
For instance, they are used to describe the alignment of crystals in
crystallography, material science and geology \cite{Bunge82, ASWK93,
  EnGoPoJu94}, the alignment of molecules and proteins in biochemistry
\cite{BaBaBeVo13} or movements in robotics \cite{ZeKuCr98} and motion tracking
\cite{RoBrBr12}.

Since, locally, the quotient manifolds $\SO/\mathcal S$ are isometric to the
rotation group $\SO$ itself all intrinsic methods for the rotation group can
be easily adapted to work on the quotients as well. Unfortunately, this is not
true for embedding based\changed{-}methods, e.g. for the interpolation methods
described in \cite{Ga18}. Explicit embeddings for the quotient manifolds
$\SO/\mathcal S$ have been investigated first by R. Arnold, P. Jupp and
H. Schaeben in \cite{ArJuSc18}. Our paper aims to extend their results by
developing a general framework for the construction of embeddings of the
quotient manifolds $\SO / \mathcal S$ that include the embeddings described in
\cite{ArJuSc18}. Our embeddings pose several nice properties, e.g. they are
all $\SO$\changed{-equivariant}\footnote{c.f. Definition
  \ref{def:equivariant}}, their \changed{images are contained in a sphere and
  the image measure $\mu \circ \mathcal E^{-1}$ induced by the rotational
  invariant measure $\mu$ on $\SO$ is centered in $\R^{d}$, i.e., has zero
  mean.} Furthermore, we find within our framework \changed{locally} isometric
embeddings of $\SO/\mathcal S$ for all finite symmetry groups $\mathcal S$ and
provide an efficient numerical method for the projection $P_{\mathcal
  M}$. \changed{The practical advantage of isometric embeddings is that
  locally isotropic methods in $\R^{d}$ translates into locally isotropic
  methods on $\SO/\mathcal S$.} 

Our paper is organized as follows. In Section~\ref{sec:general-framework} we
introduce the generic embeddings and prove in the Theorem~\ref{theorem:homo}
and Corollary~\ref{cor:sphere} that they are $\SO$\changed{-equivariant} maps
that map the quotient manifold into a subsphere of an Euclidean vector
space. Furthermore, we provide in Table \ref{tab:ualpha} the parameters such
that our embeddings coincide with the embeddings found in
\cite{ArJuSc18}. \changed{In Section~\ref{sec:rotat-invar-subsp} we
  investigate rotational invariant subspaces of $\R^{d}$ and show in
  Theorem~\ref{theorem:subspace} that the embeddings can be centered such that
  their image is contained in a linear subspace of $\R^{d}$ which allows us to
  reduce the effective dimension of the embedding. In
  Section~\ref{sec:centered-measure} we consider the rotational invariant Haar
  measure $\mu$ on $\SO$ and generalize it to a left invariant measure
  $\mu_{\mathcal S}$ on $\SO/\mathcal S$. Together with an embedding
  $\mathcal E \colon \SO/\mathcal S \to \R^{d}$ this induces an image measure
  on $\R^{d}$. In Theorem~\ref{theorem:centered} we show that the centered
  embeddings from Section~\ref{sec:centered-measure} result in centered image
  measures.} \changed{Finally}, we propose in
Section~\ref{sec:proj-onto-embedd} an iterative algorithm for the numerical
computation of the projection $P_{\mathcal M}$ of an arbitrary point in some
neighborhood of the manifold back to the manifold. To this end, we derive in
Theorem~\ref{theorem:grad} the gradient of the distance functional.


In Section \ref{sec:metric-properties} we are interested in the discrepancy
between the geodesic distance on the quotient manifold and the Euclidean
distance in the embedding. A smooth embedding into $\R^{d}$, such that the
pull back of the Euclidean metric tensor coincides with the metric tensor of
the manifold, is called isometric. According to the Nash embedding Theorem \changed{\cite{Nash54}},
there exists for every $m$-dimensional Riemannian manifold an isometric
embedding into $\R^{m(3m+11)/2}$. As all our quotient manifolds are three\changed{-}dimensional
the result guaranties the existence of an isometric embedding into
the space $\R^{30}$. It turns out that our embeddings are sufficiently
general to \changed{include locally} isometric embeddings for the quotient manifolds
$\SO/\mathcal S$ modulo all crystallographic symmetry groups $\mathcal
S$. This result is proven separately for the different types of symmetry
groups in Theorems \changed{\ref{theorem:isoC2}}, \ref{theorem:isoCk}, \ref{theorem:isoD2},
\ref{theorem:isoDk}, \ref{theorem:isoO}, \ref{theorem:isoT}. The corresponding
parameters as well as the dimension of the linear space are
summarized in Table \ref{tab:embnew}. The dimensions of the \changed{locally} isometric
embeddings vary from $8$ to $32$ depending on the symmetry group.

In the last Section \ref{sec:glob-almost-isom} we investigate the global
relationship between the geodesic distance on $\SO/\mathcal S$ and the
Euclidean distance in the embedding. According to \cite{Ve13} it is possible
to construct for each smooth and compact manifold $\mathcal M$ an embedding
$\mathcal E \colon \mathcal M \to \mathbb R^{d}$ such that the geodesic
distance on the manifold and the Euclidean distance in the embedding differ
only by a given $\varepsilon >0$, i.e.,
\begin{equation}
  \label{eq:2}
  (1-\varepsilon) \,d_{\mathcal M}(m_{1},m_{2}) \le
  d(\mathcal E(m_{1}),\mathcal E(m_{2}))
  \le (1+\varepsilon)\, d_{\mathcal M}(m_{1},m_{2}).
\end{equation}
However, the dimension $d$ of the vector space required for such an embedding
is much to large for numerical applications. In Table~\ref{tab:const} we
provide similar bounds \changed{to those} in equation~\eqref{eq:2} for the \changed{locally} isometric embeddings
defined in this paper. It turns out that locally isometric embeddings do not
necessarily lead to globally optimal bounds. Parameters for our embeddings
optimized with respect to global preservation of distances are provided in
Table~\ref{tab:cmincmax}.

\section{Embeddings of the Rotation Group}

\subsection{General Framework}
\label{sec:general-framework}
The group of rotations $\SO$ interpreted as a matrix group has a canonical
embedding $ \mathcal E \colon \SO \to \R^{9}$ given by
\begin{equation}
  \label{eq:1}
  \mathcal E (\vec R) = (\vec R \vec e_{1}, \vec R \vec e_{2}, \vec R \vec e_{3})
\end{equation}
where $\vec e_{1}$, $\vec e_{2}$, $\vec e_{3}$ is the standard basis in
$\R^{3}$. Replacing the basis vectors $\vec e_{1}$, $\vec e_{2}$, $\vec e_{3}$
by any other list of vectors $\vec u_{1}, \vec u_{2}, \ldots , \vec u_{n}$
will always result in an embedding as long as at least two of the vectors are
linearly independent. Unfortunately, this approach is not applicable to
quotients $\SO/\mathcal S$ since well definedness requires that
$\mathcal E([\vec R \vec S]_{\mathcal S}) = \mathcal E([\vec R]_{\mathcal S})$
for all symmetry operations $\vec S \in \mathcal S$. For that reason, we
generalize the embedding \eqref{eq:1} to tensor products of vectors
$\vec u_{1}, \vec u_{2}, \ldots , \vec u_{n}$. In the next definition we will
make use of the following notation. Let
$\vec \alpha = (\alpha_{1},\ldots,\alpha_{n}) \in \N^{n}$ be a
multi\changed{-}index. Then $\R^{3^\vec\alpha}$ is defined as the linear space
\begin{equation}\label{eq:vec-tensor}
  \R^{3^\vec\alpha}
  =\times_{i=1}^n \left(\otimes^{\alpha_i}\R^{3}\right)\cong \R^{\left(\sum_{i=1}^n 3^{\alpha_i}\right)}.
\end{equation}

\begin{definition}
  \label{def:E}
  Let $n \in \N$, $\vec \alpha = (\alpha_{1},\ldots,\alpha_{n}) \in \N^{n}$ \changed{be} a
  multi\changed{-}index and ${\vec u = (\vec u_{1},\ldots,\vec u_{n}) \in \R^{3n}}$ be a list
  of $n$ directions $\vec u_{j} \in \R^{3}$.  Then we define the mapping
  $\mathcal E_{\vec u}^{\vec \alpha}\colon \SO \to \R^{3^{\vec \alpha}}$ as
  \begin{align*}
    \mathcal E_{\vec u}^{\vec \alpha}(\vec R)
    = \left(\otimes^{\alpha_{1}} \vec R \vec u_{1}, \ldots,
    \otimes^{\alpha_{n}} \vec R \vec u_{n} \right).
  \end{align*}
\end{definition}

In order to define mappings that are invariant with respect to a finite
subgroup \changed{$\mathcal S < \SO$} we utilize the averaging idea.

\begin{definition}
  \label{def:ESym}
  Let \changed{$\mathcal S < \SO$} be a finite subgroup and
  $\mathcal E_{\vec u}^{\vec \alpha}\colon \SO \to \R^{3^{\vec \alpha}}$
  as defined in Definition \ref{def:E}. Then we denote by
  \begin{equation*}
    \mathcal E_{\vec u,\mathcal S}^{\vec \alpha} \colon \SO/\mathcal S \to  \R^{3^{\vec
        \alpha}}, \qquad
    \mathcal E_{\vec u,\mathcal S}^{\vec \alpha}([\vec O]_{\mathcal S})
    = \frac{1}{\abs{\mathcal S}}\sum_{\vec S \in \mathcal S} \mathcal E_{\vec u}^{\vec
      \alpha}(\vec O \vec S),
    \quad [\vec O]_{\mathcal S}\changed{=\left\{\vec O\vec R \mid \vec R\in \mathcal S\right\}} \in \SO/\mathcal S
  \end{equation*}
  its symmetrized version.
\end{definition}

In order to examine the properties of
$\mathcal E_{\vec u,\mathcal S}^{\vec \alpha}$ it we consider both, the
quotient $\SO/\mathcal S$ \changed{and} the vector space
$\R^{3^{\vec \alpha}}$ of dimension $\sum_{i=1}^n3^{\alpha_i}$ as $\SO$
manifolds equipped with the left group actions
\begin{align*}
  \vec R \triangleright [\vec O]_{\mathcal S}
  = [\vec R \vec O]_{\mathcal S}, \qquad
  \vec R \triangleright \vec v
  = (\otimes^{\vec\alpha} \vec R) \, \vec v,
\end{align*}
where $\vec R \in \SO$, $[\vec O]_{\mathcal S} \in \SO/\mathcal S$ and
$\vec v = (\vec v^{1}, \ldots, \vec v^{n}) \in \R^{3^{\vec
    \alpha}}$. \changed{The multiplication of tensor product
  $\otimes^{\vec \alpha} \vec R$
  with the tensor $\vec v \in \R^{3^{\vec \alpha}}$ is defined component-wise
  by
  $(\otimes^{\vec\alpha} \vec R)\, \vec v=\bigl((\otimes^{\alpha_i} \vec R)\,
  \vec v_i\bigr)_{i=1}^n$ and}
  \begin{equation*}
    \changed{\Bigl[(\otimes^{\alpha_{i}} \vec R)\,\vec v^{i}\Bigr]_{k_{1},\ldots,k_{\alpha_{i}}}
    = \sum_{\ell_{1}=1}^{3} \cdots \sum_{\ell_{\alpha_{i}}=1}^{3}
     R_{k_{1}\ell_{1}} \cdots R_{k_{\alpha_{i}}\ell_{\alpha_{i}}} v^{i}_{\ell_{1},\ldots,\ell_{\alpha_{i}}}.}
  \end{equation*}
  \changed{Mappings that intertwines with such group actions are called equivariant.}

\begin{definition}
  \label{def:equivariant}
\changed{  Let $G$ be a group that acts on two sets $X, Y$ via $g \triangleright x$ and
  $g \triangleright y$, $g \in G$, $x \in X$, $y \in Y$. A mapping
  $f \colon X \to Y$ is said to be an $G$-equivariant map if it intertwines
  with the group action, i.e.,}
  \begin{equation*}
   \changed{ f(g \triangleright x) = g \triangleright f(x) \quad
    \text{for all}
    \quad     g \in G, x \in X.}
  \end{equation*}
\end{definition}
\changed{It turns out that the embeddings from Definition~\ref{def:E} and
\ref{def:ESym} are indeed $\SO$-equivariant maps between the quotients
$\SO/\mathcal S$ and Euclidean vector spaces $\R^{3^{\alpha}}$.}

\begin{theorem}
  \label{theorem:homo}
  The mapping
  $\mathcal E_{\vec u,\mathcal S}^{\vec \alpha} \colon \SO/\mathcal S \to
  \R^{3^{\vec \alpha}}$ is an $\SO$\changed{-equivariant} map, i.e.,
  \begin{equation*}
    \mathcal E_{\vec u,\mathcal S}^{\vec \alpha}(\vec R \triangleright [\vec O]_{\mathcal S})
    = \vec R \triangleright \mathcal E_{\vec u,\mathcal S}^{\vec \alpha}([\vec O]_{\mathcal S})
  \end{equation*}
  for all $\vec R \in \SO$ and $[\vec O]_{\mathcal S} \in \SO/\mathcal S$.
\end{theorem}
\begin{proof}
  Let $\vec R \in \SO$ and $[\vec O]_{\mathcal S} \in \SO/\mathcal S$. Then straight forward
  computation reveals
  \begin{align*}
    \mathcal E_{\vec u,\mathcal S}^{\vec \alpha}(\vec R \triangleright [\vec O]_{\mathcal S})
    &=  \frac{1}{\abs{\mathcal S}}\sum_{\vec S \in \mathcal S} \mathcal E_{\vec u}^{\vec
      \alpha}(\vec R \vec O \vec S)\\
    &=  \frac{1}{\abs{\mathcal S}}\sum_{\vec S \in \mathcal S}
      \left(\otimes^{\alpha_{1}} \vec R \vec O \vec S \vec u_{1}, \ldots,
      \otimes^{\alpha_{n}} \vec R \vec O \vec S  \vec u_{n} \right)
     = \vec R \triangleright \mathcal E_{\vec u,\mathcal S}^{\vec \alpha}([\vec O]_{\mathcal S}).
  \end{align*}
\end{proof}

\changed{A direct consequence of
  $\mathcal E_{\vec u,\mathcal S}^{\vec \alpha}$ beeing a $\SO$-equivariant
  map is that
  $\norm{\mathcal E_{\vec u,\mathcal S}^{\vec \alpha}([\vec R]_{\mathcal S})}$
  is independent of $[\vec R]_{\mathcal S} \in \SO / \mathcal S$.}

\begin{corollary}\label{cor:sphere}
  The image
  $\mathcal E_{\vec u,\mathcal S}^{\vec \alpha}(\SO/\changed{\mathcal S})
  \subset \R^{3^{\vec \alpha}}$ is contained in a sphere \changed{with radius
  $r_{\mathcal S}$, i.e., it exists a constant $r_{\mathcal S}>0$ such that for all
  $[\vec R]_{\mathcal S} \in \SO/\mathcal S$,}
  \begin{equation*}
    \Norm{\mathcal E_{\vec u,\mathcal S}^{\vec \alpha}([\vec R]_{\mathcal S})} = r_{\mathcal S}.
  \end{equation*}
\end{corollary}
\begin{proof}
  Let $\vec R \in \SO$ be an arbitrary rotation and $\vec I \in \SO$ the
  identity. Then we have by Theorem \ref{theorem:homo} and the fact that the
  Kronecker product of orthogonal matrices is again an orthogonal matrix that
  \changed{
    \begin{align*}
      \Norm{\mathcal E_{\vec u,\mathcal S}^{\vec \alpha}([\vec R]_{\mathcal S})}^2
      &=\Bigl\langle\mathcal E_{\vec u,\mathcal S}^{\vec \alpha}([\vec R \, \vec I]_{\mathcal S}),\,
        \mathcal E_{\vec u,\mathcal S}^{\vec \alpha}([\vec R \, \vec I]_{\mathcal
        S})\Bigr\rangle
        =\Bigl\langle\vec R\triangleright
        \mathcal E_{\vec u,\mathcal S}^{\vec \alpha}( [\vec I]_{\mathcal S}),\,
        \vec R\triangleright\mathcal E_{\vec u,\mathcal S}^{\vec \alpha}([\vec I]_{\mathcal S})\Bigr\rangle\\
      &= \Bigl\langle(\otimes^{\vec \alpha}\vec R)\, \mathcal E_{\vec u,\mathcal S}^{\vec \alpha}([\vec I]_{\mathcal S}),\,
        (\otimes^{\vec \alpha}\vec R)\,\mathcal E_{\vec u,\mathcal S}^{\vec
        \alpha}([\vec I]_{\mathcal S})\Bigr\rangle
        =\Norm{\mathcal E_{\vec u,\mathcal S}^{\vec \alpha}([\vec I]_{\mathcal S})}^2.
    \end{align*}}
\end{proof}

\subsection{Rotationally Invariant Subspaces}
\label{sec:rotat-invar-subsp}

In order to prove further properties of the embeddings
$\mathcal E_{\vec u,\mathcal S}^{\vec \alpha}$ we continue by investigating
subspaces of $\R^{3^{\vec \alpha}}$ that are invariant with respect to the
group action $\triangleright$. \changed{More precisely, we search for tensors
  $\vec M_\alpha\in \R^{3^\alpha}$, such that
  $\vec R\triangleright\vec M_\alpha=\vec M_\alpha$. For $\alpha=1$ and
  $\vec v \in \R^{3}$ this means $\vec R\vec v=\vec v$ has to hold for all
  $\vec R\in \SO$. This is only fulfilled for $\vec v=\vec 0$ and, hence, the
  subspace of rotational invariant vectors in the $\R^{3}$ is just the trivial
  one.  In the case $\alpha=2$ we have for $\vec v = \vec I \in \R^{3^{2}}$
  that
  $\otimes^{2} \vec R \triangleright \vec I = \vec R \vec I \vec R^{T} = \vec
  I$ and, hence, $\vec M_{2} = \vec I$ spans a rotational invariant subspace
  of $\R^{3^{2}}$.  Indeed, we find a one-dimensional rotational invariant
  subspace for all even $\alpha$.}


\begin{lemma}\label{lem:Malpha}
  Let $\vec \alpha=(\alpha_i)_{i=1}^n$ be a multi-index. Then the
  tensor $\vec {M_{\alpha}} \in \R^{3^{\vec{\alpha}}}$ defined by
  \begin{equation*}
    (\vec M_{\alpha_i})_{j_{1},\ldots,j_{\alpha_i}}
    = \text{symm}(\otimes^{\alpha_i/2} \vec I)
    =  \frac{1}{\alpha_i!} \sum_{\sigma \in \Sigma_{\alpha_i}}
    \prod_{k=1}^{\alpha_i/2}\delta_{j_{\sigma(2k-1)},j_{\sigma(2k)}},
  \end{equation*}
  if $\alpha_i$ is even and $\vec M_{\alpha_i}=\vec 0 \in \otimes^{\alpha_i}\R^{3}$ if $\alpha_i$ is odd,  is $\SO$ invariant, i.e.,
  $ \displaystyle \; \vec R \triangleright \vec  {M_{\alpha}} = \vec {M_{\alpha}}$,
  $\vec R \in \SO$.
\end{lemma}

\begin{proof}
  For odd $\alpha_i$ there is nothing to prove. For $\vec R =\left(R_{ij}\right)_{i,j=1}^3\in \SO$ and even $\alpha \in \N_0$ we have
  \begin{align*}
    \left(\vec R \triangleright \vec M_{\alpha}\right)_{i_1,\ldots,i_\alpha}
    &=((\otimes^{\alpha} \vec R) \, \vec M_{\alpha})_{i_1,\ldots,i_\alpha}\\
    &=\sum_{j_1,\ldots,j_{\alpha}=1}^3 (\vec M_{\alpha})_{j_1,\ldots,j_{\alpha}}
      \cdot R_{i_1j_1}R_{i_2j_2}\cdots R_{i_{\alpha}j_{\alpha}}\\
    &=\frac{1}{\alpha !}\sum_{j_1,\ldots,j_\alpha=1}^3
      \left(\left(\sum_{\sigma \in \Sigma_{\alpha}}\prod_{k=1}^{\frac {\alpha}{2}}
      \delta_{j_{\sigma(2k-1)},j_{\sigma(2k)}}\right)R_{i_1j_1}R_{i_2j_2}\cdots R_{i_{\alpha}j_{\alpha}}\right)\\
    &=\frac{1}{\alpha !}\sum_{j_1,\ldots,j_\alpha=1}^3
      \left(\sum_{\sigma \in \Sigma_{\alpha}}\prod_{k=1}^{\frac{\alpha}{2}}
      \delta_{j_{\sigma(2k-1)},j_{\sigma(2k)}}\right)\prod_{l=1}^{\alpha} R_{i_lj_l}\\
    &=\frac{1}{\alpha !}\sum_{\sigma \in \Sigma_{\alpha}}\prod_{k=1}^{\frac{\alpha}{2}}
      \sum_{j_1,\ldots,j_{\alpha}=1}^3
      \delta_{j_{\sigma(2k-1)},j_{\sigma(2k)}} R_{i_{\sigma(2k-1)}j_{\sigma(2k-1)}}R_{i_{\sigma(2k)}j_{\sigma(2k)}}\\
    &=\frac{1}{\alpha !}\sum_{\sigma \in \Sigma_{\alpha}}\prod_{k=1}^{\frac {\alpha}{2}}
      \sum_{j_1,\ldots,j_{\alpha}=1}^3 R_{i_{\sigma(2k-1)}j_{\sigma(2k-1)}}R_{i_{\sigma(2k)}j_{\sigma(2k)}}.
  \end{align*}
  All the sums and products are finite, so we can interchange them. Using the
  orthogonality of $\vec R$ we obtain
  \begin{align*}
    \sum_{j_{\sigma(2k-1)}=1}^3\prod_{k=1}^{\frac {\alpha}{2}}
    R_{i_{\sigma(2k-1)}j_{\sigma(2k-1)}}R_{i_{\sigma(2k)}j_{\sigma(2k)}}
    = \langle R_{i_{\sigma(2k-1)}},R_{i_{\sigma(2k)}}\rangle
    =
      \begin{cases}
        0 & \text{if } i_{\sigma(2k-1)}\neq i_{\sigma(2k)} \\
        1 & \text{if } i_{\sigma(2k-1)}= i_{\sigma(2k)} \\
      \end{cases}
  \end{align*}
  and eventually,
  \begin{align*}
    \left(\vec R \triangleright \vec M_{\alpha}\right)_{i_1,\ldots,i_\alpha}
    &=\frac{1}{\alpha!}\sum_{\sigma \in \Sigma_{\alpha}}
      \prod_{k=1}^{\frac r2}\delta_{i_{\sigma(2k-1)},i_{\sigma(2k)}}=(\vec M_{\alpha})_{i_1,\ldots,i_{\alpha}}.
  \end{align*}
  Applying this argument element-wise for all
  $\alpha\in \{\alpha_i\}_{i=1}^n$, yields the assertion.
\end{proof}
\changed{For example in the case $\alpha =4$, the tensor $\vec M_4$ can be written as
\begin{equation*}
(\vec M_4)_{j_1,j_2,j_3,j_4}=\begin{cases} 1 & \text{if } j_1=j_2=j_3=j_4\\
\frac{1}{3} & \text{if pair-wise two indices are equal, but not all, i.e.} j_1=j_2=1, j_3=j_4=2\\
0& \text{else}
\end{cases}.
\end{equation*} }
Since, $\mathcal E_{\vec u,\mathcal S}^{\vec \alpha}$ is an $\SO$\changed{-equivariant} map, any
rotationally invariant subspace is orthogonal to the \changed{image of the} embedding
$\mathcal E_{\vec u,\mathcal S}^{\vec \alpha}(\SO)$. More precisely, we have the
following result:

\begin{lemma}\label{lem:dotpro}
\changed{  Let $\alpha\in \N$ and $\vec R\in \SO$ an arbitrary rotation. Then the inner
  product between $\mathcal E_{\vec u}^\alpha(\vec R)$ and $\vec M_{\alpha}$
  computes to}
  \begin{equation*}
    \displaystyle \dotprod{\mathcal E_{\vec u}^\alpha(\vec R)}{\vec M_{\alpha}}
    = 1.
  \end{equation*}
\end{lemma}
\begin{proof}
  We can rewrite the definition of $\mathbf M_\alpha$ for even $\alpha$ to
  \begin{align}\label{eq:Malpha}
    \vec M_\alpha
    =\frac{1}{\alpha!}\sum_{\sigma \in \Sigma_\alpha}\delta_{j_{\sigma(1)},j_{\sigma(2)}}\cdot \delta_{j_{\sigma(3)},j_{\sigma(4)}}\cdots\delta_{j_{\sigma(\alpha-1)},j_{\sigma(\alpha)}}
    =\frac{1}{\alpha!}2^{\frac {\alpha}{2}}\left(\frac {\alpha}{2}\right)!\underbrace{\left(\delta_{j_{1},j_{2}}\cdot \delta_{j_{3},j_{4}}\cdots\delta_{j_{\alpha-1},j_{\alpha}}+\ldots\right)}_{(\alpha-1)(\alpha-3)\cdots 1 \textrm{ summands}}.
  \end{align}
  The product of the $\delta$ is $1$ \changed{only}, if pairwise two $j_i$ are
  equal. Hence, we obtain the following for the scalar product if
  $\vec v=(v_1,v_2,v_3)^\top=\vec R \vec u$
  \begin{equation*}
    \dotprod{\mathcal E_{\vec u}^\alpha(\vec R)}{\vec M_{\alpha}}
    =\dotprod{\otimes^\alpha(\vec{Ru})}{\vec M_{\alpha}}=\sum_{\substack{i,j,k\\2i+2j+2k=\alpha}}a(i,j,k)v_1^{2i}v_2^{2j}v_3^{2k}
  \end{equation*}
  with coefficients $a(i,j,k)$. These coefficients have to be determined:
  \begin{align*}
    a(i,j,k)=
    &\underbrace{\frac{1}{\alpha !}2^{\frac {\alpha}{2}}\left(\frac \alpha2\right)!}_{\textrm{factor in \eqref{eq:Malpha}}}\cdot
      \underbrace{\binom{\alpha}{2i}\binom{\alpha-2i}{2j}\binom{\alpha-2i-2j}{2k}}_{\textrm{number of entries}}\\
    &\cdot\underbrace{ (2i-1)(2i-3)\cdots1\cdot(2j-1)(2j-3)\cdots1\cdot (2k-1)(2k-3)\cdots1}_{\textrm{number of summands unequal to $0$ in \eqref{eq:Malpha}} }\\
    =&\left(\frac \alpha2\right)!\cdot\frac{2^i(2i-1)(2i-3)\cdots1}{(2i)!}\cdot\frac{2^j(2j-1)(2j-3)\cdots1}{(2j)!}\cdot\frac{2^k(2k-1)(2k-3)\cdots1}{(2k)!}\\
    =&\left(\frac \alpha2\right)!\cdot\frac{1}{i!j!k!}
       =\binom{\frac \alpha2}{i,j,k}.
  \end{align*}
  With the multinomial theorem it follows that
  \begin{equation*}
    \langle \otimes ^\alpha \vec v,\vec M_\alpha\rangle=(v_1^2+v_2^2+v_3^2)^\alpha=1.
  \end{equation*}
\end{proof}

The previous lemma states that the embedded manifold is contained in the affine
subspace of all $\vec x \in \R^{3^{\vec \alpha}}$ \changed{with $\dotprod{\vec x}{\vec M_{\alpha}} = 1$.}
Next we want to shift the embedding into
the corresponding linear subspace. To this end we need to compute the
Frobenius norms $\norm{\vec M_{\alpha}}_{F}$ of the invariant tensors
$\vec M_{\alpha}$.

\begin{lemma}\label{lem:skaMr}
  Let $\alpha\in 2\mathbb{N}$. Then the Frobenius norm of the tensor $\vec M_{ \alpha}$ satisfies
  \[\norm{\vec M_\alpha}_F^2=\langle \vec M_\alpha,\vec M_ \alpha\rangle= \alpha+1.\]
\end{lemma}

\begin{proof}
  We use the formulation for the tensor $\vec M_\alpha$ from equation~\eqref{eq:Malpha}.
	Let $i_1,i_2,i_3\in\{0,1,2,\ldots,\frac \alpha 2\}$ with
  $i_1+i_2+i_3=\frac \alpha2$ such that
  \begin{align*}
    j_1,\ldots,j_{2i_1}&=1,\\
    j_{2i_1+1},\ldots,j_{2i_1+2i_2}&=2,\\
    j_{2i_1+2i_2+1},\ldots,j_{2i_1+2i_2+2i_3}&=3.
  \end{align*}
  The \changed{corresponding} entry in
  $\vec M_\alpha$ is
  \begin{align*}
    &\frac{1}{\alpha!}2^{\frac\alpha 2}\left(\frac \alpha 2\right)! \cdot (2i_1-1)(2i_1-3)\cdots 1\cdot(2i_2-1)(2i_2-3)\cdots 1 (2i_3-1)(2i_3-3)\cdots 1\\
				=&\frac{1}{\alpha!}\binom{\frac\alpha2}{i_1,i_2,i_3}\,(2i_1)!\,(2i_2)!\,(2i_3)!=\frac{(\frac \alpha2)!\,(2i_1)!\,(2i_2)!\,(2i_3)!}{\alpha!\,i_1!\,i_2!\,i_3!}.
  \end{align*}
  The values in $\vec M_\alpha$ are equal, no matter which $j_i$ are $1$ and
  similarly for $i_2$ and $i_3$. Hence, there are
  $\binom{\alpha}{2i_1,2i_2,2i_3}$ such entries in $\vec M_\alpha$. Overall we
  obtain
  \begin{align*}
    \norm{\vec M_\alpha}_F^2
    &=\sum_{\substack{i_1,i_2,i_3=0\\i_1+i_2+i_3=\frac \alpha2}}^{\frac \alpha2}\binom{\alpha}{2i_1,2i_2,2i_3}\left(\frac{(\frac \alpha 2)!\,(2i_1)!\,(2i_2)!\,(2i_3)!}{\alpha!\,i_1!\,i_2!\,i_3!}\right)^2\\
    &=\sum_{\substack{i_1,i_2,i_3=0\\i_1+i_2+i_3=\frac \alpha2}}^{\frac \alpha2}\frac{\alpha!}{(2i_1)!\,(2i_2)!\,(2i_3)!}\left(\frac{(\frac \alpha2)!^2\,(2i_1)!^2\,(2i_2)!^2\,(2i_3)!^2}{\alpha!^2\,i_1!^2\,i_2!^2\,i_3!^2}\right)\\
    &=\sum_{\substack{i_1,i_2,i_3=0\\i_1+i_2+i_3=\frac \alpha2}}^{\frac \alpha2}\left(\frac{(\frac \alpha2)!^2\,(2i_1)!\,(2i_2)!\,(2i_3)!}{\alpha!\,i_1!^2\,i_2!^2\,i_3!^2}\right)\\
    &=\frac{1}{\binom{\alpha}{\frac \alpha2}}\sum_{\substack{i_1,i_2,i_3=0\\i_1+i_2+i_3=\frac \alpha2}}^{\frac \alpha2}\binom{2i_1}{i_1}\binom{2i_2}{i_2}\binom{2i_3}{i_3}.
  \end{align*}
  With Lemma~\ref{lemma:sumBinom} follows the assertion.
\end{proof}
\changed{The previous Lemmata motivate to shift the embeddings for even $\alpha$ by a multiple of $\vec M_\alpha$ to reduce the dimension of the embedding space.}
\begin{theorem}
  \label{theorem:subspace}
  Let
  $\mathcal E_{\vec u, \mathcal S}^{\vec \alpha} \colon \SO/\mathcal S \to
  \R^{3^{\vec \alpha}}$ be the embedding defined in Definition~\ref{def:ESym}.
  Then the image of the centered embedding
  \begin{equation*}
    \tilde{\mathcal E}_{\vec u, \mathcal S}^{\vec \alpha}([\vec O]_{\mathcal S})
    = \mathcal E_{\vec u, \mathcal S}^{\vec \alpha}([\vec O]_{\mathcal S}) -
    \left(\frac{1}{\alpha_{1}+1} \vec
      M_{\alpha_{1}},\ldots,\frac{1}{\alpha_{n}+1} \vec
      M_{\alpha_{n}}\right)
  \end{equation*}
  is contained in a linear subspace of $\R^{3^\alpha}$ of dimension
  \begin{equation*}
    \sum_{i=1}^n\binom{\alpha_i+2}{\alpha_i} -
    \sum_{i=1}^n \left(\alpha_i+1\textrm{ mod } 2\right).
  \end{equation*}
\end{theorem}
\begin{proof}
  \changed{By Definition~\ref{def:E} all components $\vec T^{j} \in \R^{3^{\alpha_{j}}}$ of the
    embedding
    $\vec T = (\vec T^{1},\ldots,T^{n}) = \mathcal E_{\vec u}^{\vec
      \alpha}(\vec R)$ of an arbitrary rotation $\vec R \in \SO$ are
    symmetric tensors, i.e.,
    $\vec T^{j}_{i_1,\ldots,i_{\alpha_{j}}}=\vec
    T^{j}_{\sigma(i_1),\ldots,\sigma(i_{\alpha_{j}})}$ for any permutation $\sigma$ of $\{1,\cdots,\alpha_{j}\}$.}

  The linear space $S^\alpha(\R^3)$ of the symmetric
  $\alpha$-tensors  has the dimension
  $\binom{\alpha+2}{\alpha} $, c.f.~\cite[3.4]{CoGo08}. Thus the images
  $\mathcal E_{\vec u}^{\vec \alpha}(\SO)$ are contained in a subspace of
  $\mathbb R^{3^{\vec \alpha}}$ with dimension
  $\sum_{i=1}^n\binom{\alpha_i+2}{\alpha_i}$.  \changed{For even $\alpha$ the
    image $\mathcal E_{\vec u}^{\alpha}(\SO)$ is orthogonal to
    $\vec M_\alpha$, since for $\vec R\in \SO$}
  \begin{equation*}
    \changed{
      \langle \tilde{\mathcal E}_\vec u^\alpha(\vec R),\vec M_\alpha\rangle
      = \langle \mathcal E_\vec u^\alpha(\vec R)-\frac{1}{\alpha +1}\vec M_\alpha,\vec M_\alpha\rangle
      = \langle \mathcal E_\vec u^\alpha(\vec R),\vec M_\alpha\rangle-\frac{1}{\alpha +1}\langle\vec M_\alpha,\vec M_\alpha\rangle=0.}
    \end{equation*}

  Hence, the image $\tilde{\mathcal E}_{\vec u}^{ \alpha}(\SO)$ is contained in a
  hyperplane of the symmetric tensors in $\R^{3^\alpha}$ for every even
  component. Thus, we can reduce the dimension of every component with even
  $\alpha$ by 1. \changed{The symmetrization with the symmetry group $\mathcal S$
  does not change the dimensions.} Hence, the images
  $\tilde{\mathcal E}_{\vec u,\mathcal S}^{\vec \alpha}(\SO/\mathcal S)$ have
  dimension

  \begin{equation*}
    \sum_{i=1}^n\binom{\alpha_i+2}{\alpha_i}
    -\sum_{i=1}^n \left(\alpha_i+1\textrm{ mod } 2\right).
  \end{equation*}
\end{proof}

In \cite{ArJuSc18} the authors were especially interested in embeddings of the
rotation group modulo crystallographic point groups. These consist of the
cyclic groups $C_{k}$, and the dihedral groups $D_{k}$ with
$k \in \{1,2,3,4,6\}$, the tetrahedral group $T$ and the octahedral group $O$.
For all the corresponding quotients Table \ref{tab:ualpha} lists specific
choices of the parameters $\vec \alpha \in \R^{n}$ and
$\vec u_{1},\ldots,\vec u_{n} \in \R^{3}$ such that the generic embeddings
$\tilde{\mathcal E}_{\vec u,\mathcal S}^{\vec \alpha}$ coincide with the embeddings
reported in Table 2 of~\cite{ArJuSc18}. \changed{Here we assume the major rotational
axis in $C_{k}, D_{k}$ and $Y$ to be parallel to $\vec e_{1}$. For $O$, $T$ the
three-fold axis is assumed to be parallel to $(1,1,1)^{\top}$.}

It is important to note that at this point we have not yet proven that the
mappings $\tilde{\mathcal E}_{\vec u,\mathcal S}^{\vec \alpha}$ are indeed
embeddings, i.e., that they are injective. This will be done in
\changed{Section}~\ref{sec:isometric-embeddings}, where we shall prove that
with some modifications they are even \changed{local} isometries.

\begin{table}[h]
  \caption{Choices of the vectors $\vec u$ and the parameter $\vec\alpha$ such
  that $\tilde{\mathcal E}_{\vec u,\mathcal S}^{\vec \alpha}$ coincides with the
  embeddings reported in Table 2 of \cite{ArJuSc18}.}
  \label{tab:ualpha}
	\vspace{-0.2cm}
  \centering
  \begin{tabular}{llll}
    \hline
    $\mathcal S$  & $\vec u$& $\vec\alpha$&Dimension \\
    \hline
    $C_1$ & $(\vec e_1, \vec e_2,\vec e_3)$&(1,1,1)&9 \\
    $C_2$ & $(\vec e_1,\vec e_2)$&(1,2)&8 \\
    $C_\alpha$  $(\alpha \textrm{ even}, \alpha\geq 4)$ & $(\vec e_1,\vec e_2)$&$(1,\alpha)$&$\frac{(\alpha+2)(\alpha+1)}{2}+2$ \\
		$C_\alpha$  $(\alpha \textrm{ odd}, \alpha\geq 3)$ & $(\vec e_1,\vec e_2)$&$(1,\alpha)$&$\frac{(\alpha+2)(\alpha+1)}{2}+3$ \\
    $D_2$ & $(\vec e_1,\vec e_2)$&(2,2)&10\\
    $D_\alpha$  $(\alpha \textrm{ even}, \alpha\geq 4)$&$\vec e_1$&$\alpha$&$\frac{(\alpha+2)(\alpha+1)}{2}-1 $    \\
		$D_\alpha$  $(\alpha \textrm{ odd}, \alpha\geq 3)$&$\vec e_1$&$\alpha$&$\frac{(\alpha+2)(\alpha+1)}{2} $    \\
    $O$&$\vec e_1$&$4$&$14$\\
    $T$&$\vec e_1$&$3$&$10$\\
		\changed{$Y$}&$\vec e_1$&$10$&$66$\\
    \hline
  \end{tabular}

\end{table}

\subsection{\changed{Centered Measure}}
\label{sec:centered-measure}

\changed{
Since $\SO$ is a Lie group it can be equipped with an unique left invariant
Haar measure $\mu$. In order to define a corresponding left invariant measure
on the homogeneous space $\SO/\mathcal S$ we consider the quotient mapping}
\begin{equation*}
  \changed{
  \pi \colon \SO \to \SO/\mathcal S, \quad \pi(\vec R) = [\vec R]_{\mathcal S}}
\end{equation*}
\changed{that maps every rotation $\vec R \in \SO$ onto its coset $[\vec
R]_{\mathcal S} \in \SO/\mathcal S$. Together with the Haar measure the quotient mapping defines a left invariant measure
$\mu_{\mathcal S}$ on the quotient $\SO/\mathcal S$ via}
\begin{equation*}
\changed{  \mu_{\mathcal S}(A) = \mu(\pi^{-1}(A)),
  \quad \text{for any measurable set } A \subset \SO/\mathcal S.}
\end{equation*}
\changed{Accordingly, any embedding $\mathcal E \colon \SO/\mathcal S \to
  \R^{3^{\vec \alpha}}$ defines a push forward measure $\mathcal E \circ
  \mu_{\mathcal S}$
  on $\R^{3^{\vec \alpha}}$ via}
\begin{equation*}
  \mathcal E \circ \mu_{\mathcal S}(B) = \mu_{\mathcal S}(\mathcal E^{-1}(B)),
  \quad \text{for any measurable set } B \subset \R^{3^{\vec \alpha}}.
\end{equation*}
In the following Theorem we proof that for the centered embedding
$\tilde{\mathcal E}_{\vec u, \mathcal S}^{\vec \alpha}([\vec O]_{\mathcal S})$
the push forward measure
$\tilde{\mathcal E}_{\vec u, \mathcal S}^{\vec \alpha}([\vec O]_{\mathcal S})
\circ \mu_{\mathcal S}$ is centered in $\R^{3^{\vec \alpha}}$.


\begin{theorem}
  \label{theorem:centered}
  Let
  $\mathcal E_{\vec u, \mathcal S}^{\vec \alpha} \colon \SO/\mathcal S \to
  \R^{3^{\vec \alpha}}$ be the embedding defined in Definition~\ref{def:ESym}
  and let $\mu$ be the Haar measure on $\SO$. Then the centered
  embedding $ \tilde{\mathcal E}_{\vec u, \mathcal S}^{\vec \alpha}([\vec O]_{\mathcal S})$
  is an $\SO$\changed{-equivariant} map with
  \begin{equation*}
    \lVert\tilde {\mathcal E}_{\vec u,\mathcal S}^{\vec \alpha}([\vec
    O]_{\mathcal S})\rVert = \text{const},
    \quad [\vec O]_{\mathcal S} \in \SO/\mathcal S
  \end{equation*}
  and satisfies that the push forward measure
  $\tilde {\mathcal E}_{\vec u, \mathcal S}^\alpha\circ \mu_\mathcal S$ is centered as well, i.e.,
  its first moment satisfies
  \begin{equation*}
    \E \tilde{\mathcal E_{\vec u, \mathcal S}^\alpha} \circ\mu_\mathcal S=0.
  \end{equation*}
\end{theorem}

\begin{proof}
  The $\SO$-\changed{equivariant-map}-property follows from Theorem \ref{theorem:homo} together with Lemma~\ref{lem:Malpha}.
	For $\vec R\in \SO$ and $\vec O\in \SO/\mathcal S$ there holds
  \begin{align*}
    \tilde{\mathcal E}_{\vec u,\mathcal S}^{\vec \alpha}(\vec R \triangleright [\vec O]_{\mathcal S})
    =\mathcal E_{\vec u,\mathcal S}^{\vec \alpha}(\vec R \triangleright [\vec O]_{\mathcal S})-\vec M_{\vec\alpha}
    =\vec R \triangleright \mathcal E_{\vec u,\mathcal S}^{\vec \alpha}([\vec O]_{\mathcal S})-\vec R\triangleright\vec M_{\vec \alpha}
    = \vec R \triangleright \tilde{\mathcal E}_{\vec u,\mathcal S}^{\vec \alpha}([\vec O]_{\mathcal S}).
  \end{align*}

  Assume $\vec R$ to be distributed according to the Haar measure on $\SO$.
  Then $\vec R\vec u$ is distributed according to the spherical Borel measure
  $\sigma$ normalized to $\sigma(\mathbb S^{2}) = 1$ for any $\vec u\in
  \S^2$. For the inner products with any vector $\vec v \in \mathbb S^2$ we
  calculate
  \begin{align*}
    \langle\E (\otimes ^\alpha \vec {Ru}), \otimes ^\alpha \vec v\rangle
    &=\E\langle\otimes^\alpha \vec{Ru},\otimes^\alpha \vec v\rangle=\E((\vec{Ru})^\top \vec v)^\alpha\\
    &=\int_{\mathbb S^2}(\vec \xi^\top \vec v)^\alpha\dx\sigma(\vec \xi)
    = \begin{cases}
      \changed{0} & \text{if } \alpha \text{ odd} \\
      \changed{\frac{1}{\alpha+1}} & \text{if } \alpha \text{ even}
    \end{cases}.
  \end{align*}
  If $\alpha$ is odd, the assertion follows directly, because
  $\vec M_\alpha = \vec 0$ in this case. By Lemma \ref{lem:dotpro} we have for even $\alpha$
  \begin{align*}
    \langle\E \tilde{\mathcal E}_{\vec u}^\alpha(\vec R),\otimes^\alpha \vec v\rangle
    &=\langle\E(\otimes ^\alpha (\vec {Ru})-\frac{1}{\alpha +1} \vec M_\alpha),\otimes^\alpha \vec v\rangle
      =\E\langle\otimes ^\alpha (\vec {Ru})-\frac{1}{\alpha+1} \vec M_\alpha,\otimes^\alpha \vec v\rangle\\
    &=\E\left(((\vec {Ru})^\top \vec v)^\alpha\right)-\frac{1}{\alpha+1}=0.
  \end{align*}
	Thanks to the rotational invariance of the tensors $\vec M_{\alpha}$ the image of centered embedding is also contained in a sphere.
\end{proof}

\subsection{Projection onto the Embedding}
\label{sec:proj-onto-embedd}

A central operation of embedding\changed{-}based methods is projecting a point of the
vector space back onto the manifold. For our embeddings
${\mathcal E \colon \SO/\mathcal S \to \mathbb R^{3^{ \vec \alpha}}}$ this
means that for an arbitrary tensor
$\vec T \in \R^{3^{ \vec \alpha}}$ we ask for the rotation
${\changed{[\vec R^{*}]_\mathcal S} \in \SO/\mathcal S}$ with minimum distance
$\norm{\mathcal E(\vec R^{*}) - \vec T}$ in the embedding.
This problem has a unique solution whenever $\vec T$ is sufficiently close to the
submanifold, cf. \cite{Lee12}.

Since, by Corollary~\ref{cor:sphere}, the submanifold
$\mathcal E_{\vec u,\mathcal S}^{\vec \alpha}(\SO/\mathcal S)
\subset \R^{3^{\vec \alpha}}$ is contained in a sphere, i.e., has
constant norm, the above minimization problem is equivalent to the
maximization problem
\begin{equation}
  \label{eq:maxProblem}
    \changed{[\vec R^{*}]_\mathcal S}
    = \argmax_{\vec R \in \SO/\mathcal S} J(\vec R), \quad
    J(\vec R) = \dotprod{\mathcal E_{\vec u,\mathcal S}^{\vec
      \alpha}(\vec R)}{\vec T}.
\end{equation}
For the symmetry group $C_{1}$, i.e. no symmetry,
$\vec u = (\vec u_{1},\ldots,\vec u_{n})$,
${\vec \alpha = (1,\ldots,1) \in \R^{n}}$ and
$\vec T = (\vec T_{1},\ldots,\vec T_{n}) \in \R^{3n}$ the functional
$J \colon \SO \to \R$ simplifies to
\begin{equation*}
  J(\vec R) = \sum_{i=1}^{n} \dotprod{\vec R \vec u_{i}}{\vec T_{i}}.
\end{equation*}
An explicit formula for its maximum  is known as the Kabsch Algorithm \cite{Ho87}.

\begin{lemma}\label{lem:kabsch}
  Let $\vec u_{1},\ldots,\vec u_{n}, \vec v_{1}, \ldots, \vec v_{n} \in
  \mathbb \R^{3}$ be two lists of vectors. Then the solution of the
  maximization problem
  \begin{equation*}
    \sum_{i=1}^{n} \dotprod{\vec R \vec u_{\changed{i}}}{\vec v_{\changed{i}}} \to \max,  \quad \vec R \in \SO
  \end{equation*}
  is given by
  \begin{equation*}
    \vec R = \vec V
    \begin{pmatrix}
      1 & 0 & 0 \\
      0 & 1 & 0 \\
      0 & 0 & \det \vec V \vec U^{T} \\
    \end{pmatrix}
    \vec U^{T},
  \end{equation*}
  where $\vec U \vec \Sigma \vec V^{T} = \vec H$ is the singular value decomposition of
  the $\changed{3\times 3}$-matrix
  \begin{equation*}
    \vec H = \sum_{i=1}^{n} \vec u_{i} \otimes \vec v_{i}.
  \end{equation*}
\end{lemma}

In the case of arbitrary symmetry groups and a general embedding
$\mathcal E_{\vec u,\mathcal S}^{\vec \alpha}$ we are not able to give such a
closed form solution. For this reason, we propose to solve the maximization
problem in equation~\eqref{eq:maxProblem} numerically using a manifold
gradient method \cite{Ud94}. The next theorem provides an explicit formula for
the required gradient of $J$.

\begin{theorem}
  \label{theorem:grad}
  Let $\vec T \in \R^{3^{\alpha}}$, $\vec R \in \SO$, $\vec s$ an arbitrary
  skew\changed{-}symmetric matrix and hence, $\vec s \vec R \in T_{\vec R}\SO$ a
  tangential vector at $\vec R$. Then the gradient of $J$ in direction
  $\vec s \vec R$ is given by the inner product
  \begin{equation*}
    \nabla_{\vec s \vec R}J(\vec R)
    =\alpha \dotprod{ \vec s \triangleright_{1} (\vec
    R\triangleright \vec E)}{\vec T}
  \end{equation*}
  where $\vec E = \mathcal E_{\vec u}^\alpha (\vec I) \in \R^{3^{\alpha}}$ denotes the embedding
  of the identity matrix and $\triangleright_{1}$ denotes the multiplication
  of the matrix $\vec s$ with a tensor $\vec T \in \R^{3^{\alpha}}$ with
  respect to the first dimension of $\vec T$, i.e.,
  \begin{equation*}
    [\vec s \triangleright_{1} \vec T]_{k_{1},\ldots,k_{\alpha}}
    = \sum_{\ell_{1}=1}^{3} \vec s_{k_{1} \ell_{1}}  \vec T_{\ell_{1},k_{2},\ldots,k_{\alpha}}.
  \end{equation*}
\end{theorem}
\begin{proof}
  First of all we note that by Theorem \ref{theorem:homo} the functional $J$
  can be written as
  \begin{equation*}
    J(\vec R) = \dotprod{\vec R \triangleright \vec E}{\vec T}, \quad \vec R\in \SO.
  \end{equation*}
  Considering now a tangential vector $\vec s \vec R \in T_{\vec R} \SO$ the
  corresponding directional derivative \changed{is}
  \begin{align*}
    \nabla_{\vec s \vec R}J(\vec R)
    &=\lim_{h\rightarrow 0}\frac{1}{h}\Bigl(
      \dotprod{(\vec R+h \vec s \vec R)\triangleright \vec E}{\vec T}
      -\dotprod{\vec R\triangleright \vec E}{\vec T}\Bigr)\\
    &= \lim_{h\rightarrow 0}\frac{1}{h}
      \Bigl< \bigl(\otimes^{\alpha}(\vec R+h\vec s \vec R) - \otimes^{\alpha}\vec R\bigr)
      \vec E, \vec T \Bigr>.
  \end{align*}
  In the difference of the tensor products only the terms with $h^{1}$ remain,
  as all terms with higher power of $h$ converge to zero. Since the tensor
  $\vec E$ is symmetric the derivative simplifies further to
  \begin{align*}
    \nabla_{\vec s \vec R}J(\vec R)
    =\sum_{i=0}^{\alpha-1}
    \dotprod{\left(\otimes^i \vec R \otimes \vec s\vec R
    \otimes^{\alpha-1-i}\vec R\right) \vec E}{\vec T}
    = \alpha\dotprod{\vec s \triangleright_{1} (\vec R \triangleright \vec E)}{\vec T}.
  \end{align*}
\end{proof}
\begin{remark}
In the theorem above, we considered \changed{only} the case $\alpha\in \R$, i.e. $n=1$. For the case with multiple components, we have to sum over all components in the function
\[J(\vec R)=\sum_{i=1}^n \left\langle\mathcal E_{\vec u_i}^{\alpha_i}(\vec R),\vec T_i\right\rangle,\]
as well as in the gradient
\[ \nabla_{\vec s \vec R}J(\vec R)
    = \sum_{i=1}^n\alpha_i\dotprod{\vec s \triangleright_{1} (\vec R
      \triangleright \mathcal E_{\vec u_i}^{\alpha_i}(\vec I))}{\vec T_i}.\]
\end{remark}

\section{Distance Preservation}
\label{sec:metric-properties}

In this section we are going to investigate how well the embeddings defined in
Section \ref{sec:general-framework} preserve the geodesic distance between any
two rotations.  \changed{While the rotation group $\SO$ as a submanifold of
  $\R^{3 \times 3}$ inherits a canonical Riemanian structure it differs by the
  factor $\sqrt{2}$ from the commonly used geodesic distance
  \begin{equation}\label{eq:distSO3}
    d(\vec O_1,\vec O_2)=\arccos\left(\frac 12 \left(-1+\tr(\vec O_1^{\top}\vec O_2)\right)\right)
  \end{equation}
  on the rotation group, which has the nice interpretation of being the angle
  of rotation between the two rotations $\vec O_{1}, \vec O_{2} \in \SO$.}
For cosets $[\vec O_1]_{\mathcal S},[\vec O_2]_{\mathcal S}\in \SO/\mathcal S$
the geodesic distance~\eqref{eq:distSO3} becomes
\begin{equation}\label{eq:distSO3S}
  d([\vec O_1]_S,[\vec O_2]_S)
  = \min_{\vec R\in \mathcal S} d(\vec O_1 \vec R,\vec O_2),
\end{equation}
i.e., the minimum geodesic distance between any elements of the cosets
$[\vec O_1]_S$ and $[\vec O_2]_S$.
We first analyze this problem locally.

\subsection{\changed{Locally} Isometric Embeddings}
\label{sec:isometric-embeddings}

\changed{Let us} recall that a differentiable embedding
$\mathcal E \colon \mathcal M \to \R^{d}$ is \changed{locally} isometric if its differential
$d \mathcal E \colon T_{m} \mathcal M \to T_{\mathcal E(m)}\mathcal E(\mathcal
M)$ at each point $m \in \mathcal M$ is an isometry between vector spaces.
Since in our setting in both spaces, $\SO/\mathcal S$ and
$\R^{3^{\vec \alpha}}$, the metric is invariant with respect to the action
$\triangleright$ of $\SO$ and the embedding is an $\SO$\changed{-equivariant} map, it
suffices to prove isometry at the identity
${[\vec I]_{\mathcal S} \in \SO/\mathcal S}$ only.

In order to identify \changed{locally} isometric embeddings within our framework we need to
generalize it slightly by multiplying the components by different weights
${\vec \beta = (\beta_{1},\ldots,\beta_{n}) \in \R^{n}}$, i.e., we define
\begin{align*}
  \mathcal E_{\vec u}^{\vec \alpha,\vec \beta}(\vec R)
  = \left(\beta_1\otimes^{\alpha_{1}} \vec R \vec u_{1}, \ldots,
  \beta_n\otimes^{\alpha_{n}} \vec R \vec u_{n} \right)
\end{align*}
together with its symmetrization
\begin{equation}
  \label{eq:EbetaSym}
  \mathcal E_{\vec u,\mathcal S}^{\vec \alpha, \vec \beta} \colon
  \SO/\mathcal S \to  \R^{3^{\vec \alpha}}, \qquad
  \mathcal E_{\vec u,\mathcal S}^{\vec \alpha,\vec \beta}([\vec O]_{\mathcal S})
  = \frac{1}{\abs{\mathcal S}}\sum_{\vec S \in \mathcal S} \mathcal E_{\vec u}^{\vec \alpha,\vec \beta}(\vec O \vec S).
\end{equation}
Choosing the weights $\vec \beta$ carefully will allow us to explicitly define
\changed{locally} isometric embeddings for the quotients $\SO/\mathcal S$ of $\SO$ with
respect to all crystallographic symmetry groups.

We shall analyze the derivative
$d \mathcal E_{\vec u,\mathcal S}^{\vec \alpha, \vec \beta}([\vec I]_{\mathcal
  S}) \vec s^{(k)}$ of the embedding with respect to the following orthogonal
basis of the tangential space $T_{\vec I} \SO$ given by the
skew\changed{-}symmetric matrices
\begin{equation*}
  \vec s^{(1)} =
  \begin{pmatrix}
    0&0&0\\
    0&0&-1\\
    0&1&0
  \end{pmatrix},\quad
  \vec s^{(2)} =
  \begin{pmatrix}
    0&0&-1\\0&0&0\\1&0&0
  \end{pmatrix},\quad
  \vec s^{(3)} =
  \begin{pmatrix}
    0&-1&0\\1&0&0\\0&0&0
  \end{pmatrix}.
\end{equation*}

\changed{The basis vectors $\vec s^{(1)}$, $\vec s^{(2)}$, $\vec s^{(3)}$ are
  normalized to $\sqrt{2}$, which is exactly the factor between
   the geodesic distance defined in \eqref{eq:distSO3} and the geodesic
   distance induced by the canonical embedding. Hence, we obtain the following
   characterization on local isometry.}

\begin{lemma}
  \label{lemma:iso}
  The mapping
  $\mathcal E_{\vec u,\mathcal S}^{\vec \alpha, \vec \beta} \colon \SO/\mathcal S \to
  \R^{3^{\vec \alpha}}$ as defined in \eqref{eq:EbetaSym} is \changed{locally} isometric if
  and only if the vectors
  $d \mathcal E_{\vec u,\mathcal S}^{\vec \alpha, \vec \beta}([\vec I]_{\mathcal S}) \vec
  s^{(k)}$ are orthonormal in $\R^{3^\vec \alpha}$.
\end{lemma}
\begin{proof}
  The mapping
  $d \mathcal E_{\vec u,\mathcal S}^{\vec \alpha, \vec \beta}([\vec I]_{\mathcal S})$
  is linear and $\{\vec s^{(k)}\}_{k=1}^3$ is a basis in $T_{\vec
    I}\SO$. Hence,
  $\mathcal E_{\vec u,\mathcal S}^{\vec \alpha, \vec \beta}$ is \changed{locally} isometric if
  and only if the vectors
  $d \mathcal E_{\vec u,\mathcal S}^{\vec \alpha, \vec \beta}([\vec I]_{\mathcal S})
  \vec s^{(k)}$ are orthonormal in the tangent space
  $T_{\mathcal E(\vec I)}\R^{3^\vec \alpha}$.
\end{proof}

For the differential of the mapping
$\mathcal E_{\vec u,\mathcal S}^{\vec \alpha \vec \beta}$ we have the
following lemma.

\begin{lemma}\label{lem:vecpro}
  Let $\alpha \in \mathbb{N}$, $\vec u \in \mathbb S^{2}$ \changed{be} an arbitrary direction and
  $\vec s \in T_{\vec I}\SO$ \changed{be} an arbitrary skew-symmetric matrix. Then
  \begin{equation*}
    d \mathcal E_{\vec u}^{\alpha}(\vec I) \vec s
    = \sum_{i=0}^{\alpha-1} \left(\otimes^{i} \vec u \right) \otimes \vec s \vec u
    \otimes \left(\otimes^{\alpha-i-1} \vec u\right).
  \end{equation*}
\end{lemma}
\begin{proof}
  Let $\gamma(t)$ be a curve in $\SO$ such that $\dot{\gamma}(0)=\vec s$ and
  $\gamma(0)=\vec I$.  The image of the map
  $d \mathcal E_{\vec u,\mathcal S}^{\vec \alpha, \vec \beta}([\vec
  I]_{\mathcal S})$ of $\vec s$ is given by
  \begin{equation*}
    d \mathcal E_{\vec u}^{\alpha}(\vec I) \vec s
    = \frac{\dx}{\dx t}\left(\otimes^\alpha \left(\gamma(t)\cdot \vec u\right)\right)\bigg|_{t = 0}.
  \end{equation*}
  With the chain rule it follows
  \begin{align*}
    d \mathcal E_{\vec u}^{\alpha}(\vec I) \vec s
    &= \sum_{i=0}^{\alpha-1} \left(\otimes^i(\gamma(t)\cdot \vec u )\otimes \dot{\gamma}(t)\vec u\otimes(\otimes ^{\alpha-i-1}\gamma(t)\vec u)\right)\bigg|_{t = 0}\\
    &=\sum_{i=0}^{\alpha-1} \left(\otimes^{i} \vec u \right) \otimes \vec s \vec u
      \otimes \left(\otimes^{\alpha-i-1} \vec u\right).
  \end{align*}
\end{proof}

In the following we will find \changed{locally} isometric embeddings for all crystallographic
symmetry groups. Therefore, we will proceed as follows. First we consider the
cyclic groups $C_{k}$, $k \in \N$, followed by the dihedral groups $D_{k}$,
$k \in \N$ and finally the \changed{tetrahedral} group $T$, the \changed{octahedral} group $O$ \changed{and the icosahedral group $Y$}. The
parameters for these locally isometric embeddings are summarized in
Table~\ref{tab:embnew}. \changed{The differences to the embeddings in \cite{ArJuSc18} are marked in magenta.}
 For the cyclic and the dihedral groups we assume the
major rotational axis to be aligned in $\vec e_{1}$--direction and the two\changed{-}fold
axis parallel to $\vec e_{2}$.

\changed{For the symmetry group $C_1$ the canonical embedding \eqref{eq:1} of $\SO$ in
$\R^{3 \times 3}$ is by definition locally isometric, up to the factor $\sqrt{2}$,
so multiplication of all components with $\sqrt{2}^{-1}$ leads to local isometry.}
The symmetry group $C_2$ is a special case, because in
contrast to $C_k$ for $k>2$ the vectors $\vec{Oe}_2$ for
$\vec O\in C_k$ do not span the plane orthogonal to $\vec e_1$. For this
reason we need to add an additional component in contrast to the embedding in
\cite{ArJuSc18}.

\begin{theorem}
  \label{theorem:isoC2}
  Let $\vec u = (\vec e_{1},\vec e_{2},\vec e_{3})$, $\vec \alpha = (1,2,2)$
  and $\vec \beta = (\frac{1}{\sqrt{2}},\frac{1}{2},\frac{1}{2} )$. Then
  $\mathcal E_{\vec u,C_{2}}^{\vec \alpha, \vec \beta}$ is a \changed{locally} isometric
  embedding.
\end{theorem}
\begin{proof}
  There holds
  \begin{align*}
    d \mathcal E_{\vec u,C_2}^{\vec \alpha, \vec \beta}([\vec I]_{C_2}) \vec s^{(1)}
    &=\left(\beta_1\,
      \begin{pmatrix}
        0\\0\\0
      \end{pmatrix},\beta_2\,
    \begin{pmatrix}
      0&0&0\\0&0&1\\0&1&0
    \end{pmatrix},\beta_3\,
                         \begin{pmatrix}0&0&0\\0&0&-1\\0&-1&0\end{pmatrix}\right),\\
    d \mathcal E_{\vec u,C_2}^{\vec \alpha, \vec \beta}([\vec I]_{C_2}) \vec
    s^{(2)}&=\left(\beta_1\,\begin{pmatrix}0\\0\\1 \end{pmatrix},\beta_2\,\begin{pmatrix}0&0&0\\0&0&0\\0&0&0\end{pmatrix},\beta_3\,\begin{pmatrix}0&0&-1\\0&0&0\\-1&0&0\end{pmatrix}\right),\\
    d \mathcal E_{\vec u,C_2}^{\vec \alpha, \vec \beta}([\vec I]_{C_2}) \vec
    s^{(3)}&=\left(\beta_1\,\begin{pmatrix}0\\1\\0 \end{pmatrix},\beta_2\,\begin{pmatrix}0&-1&0\\-1&0&0\\0&0&0\end{pmatrix},\beta_3\,\begin{pmatrix}0&0&0\\0&0&0\\0&0&0\end{pmatrix}\right).\\
  \end{align*}
  These three vectors are orthogonal. To normalize them, we have to solve
  \begin{align*}
   2\beta_2^2+2\beta_3^2= \beta_1^2+2\beta_3^2=\beta_1^2+2\beta_2^2=1,
  \end{align*}
  which yields $\beta_1=\frac{1}{\sqrt 2},\beta_2=\beta_3=\frac{1}{2}$.
\end{proof}

For the symmetry groups $C_k$ for $k>2$ we first show the orthogonality of the tangent vectors $d\mathcal E_{\vec u,C_k}^{\vec\alpha}([\vec I]_{C_k})$.

\begin{lemma}\label{lem:orthCk}
Let $k \in \N$ with $k>2$, $\vec u = (\vec e_{1},\vec e_{2})$ and
  $\vec \alpha = (1,k)$. Then the vectors $d\mathcal E_{\vec u,C_{k}}^{\vec \alpha}$ are orthogonal.
\end{lemma}
\begin{proof}
  For the rank one component $d \mathcal E_{\vec e_1,C_k}^{1, \beta_1}([\vec
  I]_{C_k})$ of $d \mathcal E_{\vec u,C_k}^{\vec \alpha, \vec \beta}([\vec
  I]_{C_k})$ orthogonality follows from
  \begin{equation}
    \changed{\label{eq:sca1}
    d \mathcal E_{\vec e_1,C_k}^{1, \beta_1}([\vec I]_{C_k}) \vec  s^{(1)} =
    \vec 0, \quad
    d \mathcal E_{\vec e_1,C_k}^{1, \beta_1}([\vec I]_{C_k}) \vec  s^{(2)} =
    \beta_{1} \vec e_{3}, \quad
    d \mathcal E_{\vec e_1,C_k}^{1, \beta_1}([\vec I]_{C_k}) \vec  s^{(3)} =
    \beta_{1} \vec e_{2}.}
  \end{equation}
  For the rank $k$ component $d \mathcal E_{\vec e_2,C_k}^{k, \beta_2}([\vec
  I]_{C_k})$ we use the Lemma~\ref{lem:vecpro} and define for $\ell=1,2,3$
  \begin{align}\label{eq:defBl}
   \vec B_\ell
    := d \mathcal E_{\vec e_2,C_k}^{k}([\vec I]_{C_k}) \vec  s^{(\ell)}
    = \sum_{i=0}^{k-1} \frac 1k\sum_{j=0}^{i-1} \left(\otimes^{i} \vec v_j \right) \otimes \vec  s^{(\ell)}\vec v_j
    \otimes \left(\otimes^{k-i-1} \vec v_j\right),
  \end{align}
  where the vectors
  $\vec v_j= (0, \cos \frac{2\pi j}{k}, \sin \frac{2\pi j}{k})^{\top}$ result
  from applying all symmetries from $C_{k}$ to $\vec e_{2}$. The inner products
  between these rank $k$ tensors $ \vec B_\ell$, $\ell=1,\ldots,3$ \changed{are}
  \begin{align}\label{eq:scaproten2}
    \begin{split}
      \left\langle\vec B_{\ell_1},\vec B_{\ell_2}\right\rangle
      &=
      \frac{k(k-1)}{k^2}\sum_{j_1=0}^{k-1}\sum_{j_2=0}^{k-1}\left\langle\vec
        v_{j_1},\vec v_{j_2}\right\rangle^{k-2}
      \left\langle\vec s^{(\ell_1)}\vec v_{j_1},\vec v_{j_2}\right\rangle\left\langle\vec s^{(\ell_2)}\vec v_{j_2},\vec v_{j_1}\right\rangle\\
      &\quad+\frac{k}{k^2}\sum_{j_1=0}^{k-1}\sum_{j_2=0}^{k-1}\left\langle\vec v_{j_1},\vec v_{j_2}\right\rangle^{k-1}\left\langle\vec s^{(\ell_1)}\vec v_{j_1},\vec s^{(\ell_2)}\vec v_{j_2}\right\rangle.
    \end{split}
  \end{align}
  Using
  \begin{align*}
    \vec s^{(1)}\vec v_j=\begin{pmatrix} 0\\ -\sin\frac{2\pi j}{k}\\ \cos\frac{2\pi j}{k}\end{pmatrix},\quad
    \vec s^{(2)}\vec v_j=\begin{pmatrix}  -\sin\frac{2\pi j}{k}\\ 0\\0\end{pmatrix},\quad
    \vec s^{(3)}\vec v_j=\begin{pmatrix}  -\cos\frac{2\pi j}{k}\\ 0\\0\end{pmatrix}
  \end{align*}
  we observe for all $j_1,j_2$ and $\ell=2,3$ the orthogonality
  $\left\langle\vec s^{(\ell)}\vec v_{j_1},\vec v_{j_2}\right\rangle = 0$ and hence, the first double sum in \eqref{eq:scaproten2} is zero whenever $\ell_{1}
  \ne \ell_{2}$.

  In the second double sum we have
  $\left\langle\vec s^{(\ell_1)}\vec v_{j_1},\vec s^{(\ell_2)}\vec
    v_{j_2}\right\rangle = 0$ for all $\ell_{1} \ne \ell_{2}$ except for the pair
  $\ell_{1},\ell_{2} \in \{2,3\}$. For this specific case we \changed{use the calculation in \eqref{eq:A1} and get}
  \begin{align*}
    \left\langle\vec B_{2},\vec B_{3}\right\rangle
    &=\sum_{j_1,j_{2}=0}^{k-1} \left\langle\vec v_{j_1},\vec
      v_{j_2}\right\rangle^{k-1}\left\langle\vec s^{(2)}\vec v_{j_1},\vec
      s^{(3)}\vec v_{j_2}\right\rangle\\
    &= \sum_{j_1,j_{2}=0}^{k-1} \cos^{k-1} \tfrac{2\pi(j_{1}-j_{2})}{k}
      \sin \tfrac{2\pi j_{1}}{k} \cos \tfrac{2\pi j_{2}}{k} \\
			&=0.
  \end{align*}
\end{proof}

In order to prove
  $\norm{d \mathcal E_{\vec u,C_k}^{\vec \alpha, \vec \beta}([\vec I]_{C_k})
    \vec s^{(k)}} = 1$ we continue by calculating
  $\norm {\vec B_\ell}^2=\left\langle\vec B_\ell,\vec B_\ell\right\rangle$ for $\ell=1,2,3$.

\begin{lemma}\label{lem:normBl}
For the tensors $\vec B_\ell$ defined in equation \eqref{eq:defBl} we have
\begin{align*}
\norm{\vec B_1}^2&=\begin{cases}
\frac {k^2}{2^{k-1}} &\text{if } k \text{ odd}\\
-\frac{k(k-1)}{2^{k-2}}\binom{k-2}{\frac k2-1}+\frac {k^2}{2^{k}}\left(\binom{k}{\frac k2}+2\right) &\text{if } k \text{ even}
\end{cases},\\
\norm{\vec B_2}^2=\norm{\vec B_3}^2&=\begin{cases}
\frac{k}{2^{k}}&\text{if } k \text{ is odd}\\
\frac{k}{2^{k+1}}\left(2+\binom{k-1}{\frac k2}+\binom{k-1}{\frac k2-1}\right)&\text{if } k \text{ is even}
\end{cases}.
\end{align*}
\end{lemma}
\begin{proof}

By equation~\eqref{eq:scaproten2} \changed{and the calculations in \eqref{eq:A2}} we obtain
\begin{align*}
\norm{\vec B_{1}}^2 &= \frac{(k-1)}{k}\sum_{j_1=0}^{k-1}\sum_{j_2=0}^{k-1}\left\langle\vec v_{j_1},\vec v_{j_2}\right\rangle^{k-2}\left\langle\vec s^{(1)}\vec v_{j_1},\vec v_{j_2}\right\rangle\left\langle\vec s^{(1)}\vec v_{j_2},\vec v_{j_1}\right\rangle\\
&\quad+\frac{1}{k}\sum_{j_1=0}^{k-1}\sum_{j_2=0}^{k-1}\left\langle\vec v_{j_1},\vec v_{j_2}\right\rangle^{k-1}\left\langle\vec s^{(1)}\vec v_{j_1},\vec s^{(1)}\vec v_{j_2}\right\rangle\\
&= -\frac{(k-1)}{k}\sum_{j_1=0}^{k-1}\sum_{j_2=0}^{k-1}\cos^{k-2}\left(\frac{2\pi(j_1-j_2)}{k}\right)\sin^2\left(\frac{2\pi(j_1-j_2)}{k}\right)\\
&\quad+\frac{1}{k}\sum_{j_1=0}^{k-1}\sum_{j_2=0}^{k-1}\cos^{k-1}\left(\frac{2\pi(j_1-j_2)}{k}\right)\cos\left(\frac{2\pi(j_1-j_2)}{k}\right)\\
&=\begin{cases}
\frac {k^2}{2^{k-1}} &\text{if } k \text{ odd}\\
-\frac{k(k-1)}{2^{k-2}}\binom{k-2}{\frac k2-1}+\frac {k^2}{2^{k}}\left(\binom{k}{\frac k2}+2\right) &\text{if } k \text{ even}
\end{cases}.
\end{align*}

Next we investigate the tensor $\vec B_3$. \changed{With the calculations in \eqref{eq:A3} in the appendix we get the following.}
\begin{align*}
\norm {\vec B_3}^2&=\frac{1}{k}\sum_{j_1=0}^{k-1}\sum_{j_2=0}^{k-1}\left\langle\vec v_{j_1},\vec v_{j_2}\right\rangle^{k-1}\left\langle\vec s^{(3)}\vec v_{j_1},\vec s^{(3)}\vec v_{j_2}\right\rangle\\
&=\frac{1}{k}\sum_{j_1=0}^{k-1}\sum_{j_2=0}^{k-1}\cos^{k-1}\left(\frac{2\pi(j_1-j_2)}{k}\right)\cos\left(\frac{2\pi j_1}{k}\right)\cos\left(\frac{2\pi j_2}{k}\right)\\
&=\begin{cases}
\frac{k}{2^{k}}&\text{if } k \text{ is odd}\\
\frac{k}{2^{k+1}}\left(2+\binom{k-1}{\frac k2}+\binom{k-1}{\frac k2-1}\right)&\text{if } k \text{ is even}
\end{cases}.
\end{align*}
For the norm $\norm{\vec B_2}^2$ we only have to change some signs in the previous calculation and receive in the end ${\norm{\vec B_2}^2=\norm{\vec B_3}^2}$.
\end{proof}
Summarizing these Lemmata we find weights $\vec \beta$ for all crystallographic
symmetry groups $\mathcal S$ such that the corresponding embeddings are isometries.

\begin{theorem}
  \label{theorem:isoCk}
  Let $k \in \N$ with $k>2$, $\vec u = (\vec e_{1},\vec e_{2})$ and
  $\vec \alpha = (1,k)$. Then the embeddings
  $\mathcal E_{\vec u,C_{k}}^{\vec \alpha, \vec \beta}$ with the factors
  \begin{equation*}
    \vec \beta = \left(\sqrt{1-\frac{\norm{\vec B_2}^2}{\norm{\vec B_1}^2}},\frac {1}{\norm{\vec B_1}}\right)^\top
  \end{equation*}
  with the norms from Lemma~\ref{lem:normBl} are \changed{locally} isometric embeddings. The concrete factors for $k=3,4,6$ are listed in Table~\ref{tab:embnew}.
\end{theorem}
\begin{proof}
  We use equation \eqref{eq:sca1} for the rank 1 tensor. To normalize the
  vectors
  $d \mathcal E_{\vec u,C_k}^{\vec\alpha, \vec\beta}([\vec I]_{C_k}) \vec
  s^{(\ell)}$ for $\ell=1,2,3$ we have to solve for every $k$ equations of the form
\begin{align*}
\beta_2^2\cdot \norm{\vec B_1}^2=\beta_1^2+\beta_2^2\cdot\norm{\vec B_2}^2= \beta_1^2+\beta_2^2\cdot\norm{\vec B_3}^2=1,
\end{align*}
which always has a solution since $\norm{B_2}=\norm{B_3}$. We receive the positive solution by
\begin{align*}
\beta_1=\sqrt{1-\frac{\norm{\vec B_2}^2}{\norm{\vec B_1}^2}},\quad
\beta_2=\frac {1}{\norm{\vec B_1}}.
\end{align*}
\end{proof}

\begin{table}[ht]
  \caption{Choices of the vectors $\vec u$ and the parameters $\vec\alpha,\vec \beta$ such
  that the embeddings $\tilde{\mathcal E}_{\vec u,\mathcal S}^{\vec \alpha, \vec \beta}$ are locally isometric.}
  \label{tab:embnew}
	\vspace{-0.2cm}
  \centering
  \begin{tabular}{llllc}
    \hline
     $\mathcal S$   & $\vec u$& $\vec\alpha$&$\color{magenta}{\vec\beta}$&Dimension \\
    \hline
    $C_1$ & $(\vec e_1, \vec e_2,\vec e_3)$	&	$(1, 1, 1)$	&	$(\frac{1}{\sqrt 2},\frac{1}{\sqrt 2},\frac{1}{\sqrt 2})$ &	$9$ \\
    $C_2$ & $(\vec e_1,\vec e_2,\color{magenta}{\vec e_3}\color{black})$&(1, 2, \color{magenta}{2}\color{black})	&	$\left(\frac{1}{\sqrt 2},\frac{1}{2},\frac{1}{2}\right)$	&	$\color{magenta}{13}$ \\
    $C_3$ & $(\vec e_1,\vec e_2)$	&	$(1, 3)$	&	$\left(\sqrt{\frac{5}{6}},\frac{\sqrt 4}{3}\right)$	& $13$\\
    $C_4$	&	$(\vec e_1,\vec e_2)$	&	$(1, 4)$	&	$(\frac{1}{\sqrt 2},\frac{1}{\sqrt 2})$ & $17$\\
    $C_6$	&	$(\vec e_1,\vec e_2)$	&	$(1, 6)$	&	$\left(\frac{1}{\sqrt 12},\frac{2\,\sqrt{2}}{3}\right)$	& $30$\\
    $D_2$ & $(\vec e_1,\vec e_2,\color{magenta}{\vec e_3}\color{black})$&$(2, 2, \color{magenta}{2}\color{black})$	&	$\left(\frac{1}{2},\frac{1}{2},\frac{1}{2}\right)$	&	$\color{magenta}15$\\
    $D_3$ & $(\color{magenta}{\vec e_1}\color{black},\vec e_2)$	&	$(\color{magenta}{2}\color{black}, 3)$	&	$\left(\sqrt{\frac{5}{12}},\frac{\sqrt 4}{3}\right)$	&  $\color{magenta}15$\\
    $D_4$ &	$(\color{magenta}{\vec e_1}\color{black},\vec e_2)$	&	$(\color{magenta}{2}\color{black}, 4)$	&	$\left(\frac{1}{2},\frac{1}{\sqrt 2}\right)$	&	$\color{magenta}19$\\
    $D_6$ &	$(\color{magenta}{\vec e_1}\color{black},\vec e_2)$	&	$(\color{magenta}{2}\color{black}, 6)$	&	$\left(\frac{1}{\sqrt{ 24}},\frac{2\sqrt{2}}{3}\right)$	& $\color{magenta}32$\\
    $O$		&	$\vec e_1$	&	$4$	&	$\frac{3}{2\sqrt{2}}$	&	$14$\\
    $T$		&	$\vec e_1$	&	$3$	&	$\frac{3}{2\sqrt{2}}$	&	$10$\\
		\changed{$Y$}		&	$\vec e_1$	&	$10$	&	$\frac{75}{8\sqrt{95}}$	&	$66$\\
    \hline
  \end{tabular}

\end{table}

For the symmetry groups $D_k$ the case $k=2$ is a special case for the same reasons as $C_2$.
\begin{theorem}
  \label{theorem:isoD2}
Let $\vec u = (\vec e_{1},\vec e_{2},\vec e_{2})$, $\vec \alpha = (2,2,2)$
  and $\vec \beta = (\frac{1}{2},\frac{1}{2},\frac{1}{2} )$. Then
  $\mathcal E_{\vec u,D_{2}}^{\vec \alpha, \vec \beta}$ is an \changed{locally} isometric
  embedding.
\end{theorem}
\begin{proof}
The second and third component are the same \changed{as in} the case $C_2$. \changed{Analogously} to this case we have to solve
\begin{align*}
2\,\beta_1^2+2\,\beta_2^2=2\,\beta_2^2+2\,\beta_3^2=1,
\end{align*}
which yields $\beta_1=\beta_2=\beta_3=\frac{1}{2}$.
\end{proof}

\begin{theorem}
  \label{theorem:isoDk}
 Let $k \in \N$ with $k>2$, $\vec u = (\vec e_{1},\vec e_{2})$ and $\vec \alpha = (2,k)$.
Then there exist factors $\vec \beta $, such that $\mathcal E_{\vec u,D_{k}}^{\vec \alpha, \vec \beta}$ is an locally isometric
  embedding.
\end{theorem}
\begin{proof}
As in the case $C_k$ we get the same second components $B_1,B_2$ and $B_3$. Only
the first component is now a $3\times 3$-matrix and not just a vector. The three
vectors $d \mathcal E([\vec I]_{D_k})\vec s^{(\ell)}$ are again orthogonal. For the normalization we have to solve
\begin{align*}
\beta_2^2\cdot \norm{B_1}^2=2\,\beta_1^2+\beta_2^2\cdot\norm{B_2}^2= 2\,\beta_1^2+\beta_2^2\cdot\norm{B_3}^2=1,
\end{align*}
which yields the same solutions for $\beta_2$ as in the case $C_k$, but
for $\beta_1$ we have to divide the solution from $C_k$ by $\sqrt{2}$.
\end{proof}

For the cubic symmetry group the \changed{locally} isometric embedding requires only a single
vector. More precisely, we have the following result.

\begin{theorem}
  \label{theorem:isoO}
  Let $\vec u = \vec e_{1}$, $\vec \alpha = 4$
  and $\vec \beta = \frac{3}{2\sqrt{2}} $. Then
  $\mathcal E_{\vec u,O}^{\vec \alpha, \vec \beta}$ is a \changed{locally} isometric
  embedding.
\end{theorem}
\begin{proof}
  The vectors $\vec R\vec e_1$ for $\vec R\in O$ are in the set
  $\{\pm \vec e_1,\pm \vec e_2,\pm \vec e_3\}$. Since
  ${\otimes^4 \vec x=\otimes^4 (-\vec x)}$, we only have to consider the three
  vectors $\vec v_i=\vec e_i$ for $i=1,2,3$. With respect to the skew\changed{-}symmetric
	basis $\vec s^{(k)}$, $k=1,2,3$ we obtain
  \begin{align*}
    \vec s^{(1)} \vec v_1&=\vec 0,\quad &\vec s^{(1)} \vec v_2&=\vec e_3,\quad &\vec s^{(1)} \vec v_3&=-\vec e_2,\\
    \vec s^{(2)} \vec v_1&=\vec e_3,\quad &\vec s^{(2)} \vec v_2&=\vec 0,\quad &\vec s^{(2)} \vec v_3&=-\vec e_1,\\
    \vec s^{(3)} \vec v_1&=\vec e_2,\quad &\vec s^{(3)} \vec v_2&=-\vec e_1,\quad &\vec s^{(3)} \vec v_3&=\vec 0.
  \end{align*}
  By Lemma~\ref{lem:vecpro} the scalar products in the embedding \changed{are}
  \begin{align*}
    \begin{split}
      \left\langle d E_{\vec u,O}^{\vec \alpha, \vec \beta}\vec s^{(\ell_1)},
        dE_{\vec u,O}^{\vec \alpha, \vec \beta}\vec s^{(\ell_2)}\right\rangle
      &= \frac{4\cdot 3}{3^2}\sum_{j_1=1}^{3}\sum_{j_2=1}^{3}
      \left\langle\vec v_{j_1},\vec v_{j_2}\right\rangle^{2}\left\langle\vec s^{(\ell_1)}\vec v_{j_1},\vec v_{j_2}\right\rangle\left\langle\vec s^{(\ell_2)}\vec v_{j_2},\vec v_{j_1}\right\rangle\\
      &\quad+\frac{4}{3^2}\sum_{j_1=1}^{3}\sum_{j_2=1}^{3}
      \left\langle\vec v_{j_1},\vec v_{j_2}\right\rangle^{3}\left\langle\vec s^{(\ell_1)}\vec v_{j_1},\vec s^{(\ell_2)}\vec v_{j_2}\right\rangle
    \end{split}\\
      &=\frac{4\cdot 3}{3^2}\sum_{j=1}^{3}
        \left\langle\vec s^{(\ell_1)}\vec v_{j},\vec v_{j}\right\rangle\left\langle\vec s^{(\ell_2)}\vec v_{j},\vec v_{j}\right\rangle
        +\frac{4}{3^2}\sum_{j=1}^{3}\left\langle\vec s^{(\ell_1)}\vec v_{j},\vec s^{(l_2)}\vec v_{j}\right\rangle\\
      &=\frac{4}{3^2}\sum_{j=1}^{3}\left\langle\vec s^{(\ell_1)}\vec v_{j},\vec s^{(\ell_2)}\vec v_{j}\right\rangle=\frac{8}{9}\,\delta_{\ell_1,\ell_2}.
  \end{align*}
  Hence, the tangential vectors are orthogonal and normalized for $\beta_1=\frac{ 3}{2\sqrt 2}$.
\end{proof}

The tetrahedral symmetry $T$ also requires only one component, so we have the following result.

\begin{theorem}
  \label{theorem:isoT}
  Let $\vec u = \frac{1}{\sqrt 3}\begin{pmatrix}1 \\ 1\\1 \end{pmatrix}$,
  $\vec \alpha = 3$ and $\vec \beta =\frac{ 3}{2\sqrt 2} $. Then
  $\mathcal E_{\vec u,T}^{\vec \alpha, \vec \beta}$ is a \changed{locally} isometric embedding.
\end{theorem}

\begin{proof}
  The vectors $\vec R\vec u_1$ for $\vec R\in T$ are
  \begin{equation*}
    \vec v_1=\frac{1}{\sqrt 3}
    \begin{pmatrix}
      1\\1\\1
    \end{pmatrix},\quad
    \vec v_2=\frac{1}{\sqrt 3}
    \begin{pmatrix}
      -1\\-1\\1
    \end{pmatrix},\quad
    \vec v_3=\frac{1}{\sqrt 3}
    \begin{pmatrix}
      -1\\1\\-1
    \end{pmatrix},\quad
    \vec v_4=\frac{1}{\sqrt 3}
    \begin{pmatrix}
      1\\-1\\-1
    \end{pmatrix}
  \end{equation*}
  and satisfy $\langle \vec v_i,\vec v_j\rangle=-\frac 13 $ for
  $i\neq j$. By Lemma~\ref{lem:vecpro} we have
  \begin{equation*}
    d \mathcal E_{\vec u,T}^{\alpha}(\vec I) \vec s^{(\ell)}
    = \sum_{j=1}^4\sum_{i=0}^{2} \left(\otimes^{i} \vec v_j \right) \otimes \vec s^{(\ell)} \vec v_j
    \otimes \left(\otimes^{2-i} \vec v_j\right)
  \end{equation*}
  and hence, the scalar products of the basis vectors \changed{are}
  \begin{align*}
    \begin{split}
      \left\langle d E_{\vec u,T}^{\vec \alpha}\vec s^{(\ell_1)},
        dE_{\vec u,T}^{\vec \alpha}\vec s^{(\ell_2)}\right\rangle
      &= \frac{3\cdot 2}{4^2}\sum_{j_1=1}^{4}\sum_{j_2=1}^{4}
      \left\langle\vec v_{j_1},\vec v_{j_2}\right\rangle
      \left\langle\vec s^{(\ell_1)}\vec v_{j_1},\vec v_{j_2}\right\rangle
      \left\langle\vec s^{(\ell_2)}\vec v_{j_2},\vec v_{j_1}\right\rangle\\
      &\quad+\frac{3}{4^2}\sum_{j_1=1}^{4}\sum_{j_2=1}^{4}
      \left\langle\vec v_{j_1},\vec v_{j_2}\right\rangle^2
      \left\langle\vec s^{(\ell_1)}\vec v_{j_1},\vec s^{(\ell_2)}\vec v_{j_2}\right\rangle.
    \end{split}
  \end{align*}
  Using the symmetry of vectors $\vec v_j$ and $\vec s^{(l)}\vec v_j$
  \begin{align*}
    \vec s^{(1)}v_1
    &=\frac{1}{\sqrt 3}
      \begin{pmatrix}
        0\\-1\\1
      \end{pmatrix},\quad
    &\vec s^{(1)}\vec v_2
    &=\frac{1}{\sqrt 3}\begin{pmatrix}0\\-1\\-1\end{pmatrix},\quad
    &\vec s^{(1)}\vec v_3
    &=\frac{1}{\sqrt 3}\begin{pmatrix}0\\1\\1\end{pmatrix},\quad
    &\vec s^{(1)}\vec v_4
    &=\frac{1}{\sqrt 3}\begin{pmatrix}0\\1\\-1\end{pmatrix}\\
    \vec s^{(2)}\vec v_1
    &=\frac{1}{\sqrt 3}\begin{pmatrix}-1\\0\\1\end{pmatrix},\quad
    &\vec s^{(2)}\vec v_2&=\frac{1}{\sqrt
                      3}\begin{pmatrix}-1\\0\\-1\end{pmatrix},\quad
    &\vec s^{(2)}\vec v_3&=\frac{1}{\sqrt
                      3}\begin{pmatrix}1\\0\\-1\end{pmatrix},\quad
    &\vec s^{(2)}\vec v_4&=\frac{1}{\sqrt 3}\begin{pmatrix}1\\0\\1\end{pmatrix}\\
    \vec s^{(3)}\vec v_1&=\frac{1}{\sqrt 3}\begin{pmatrix}-1\\1\\0\end{pmatrix},\quad
    &\vec s^{(3)}\vec v_2
    &=\frac{1}{\sqrt 3}\begin{pmatrix}1\\-1\\0\end{pmatrix},\quad
    &\vec s^{(3)}\vec v_3&=\frac{1}{\sqrt 3}\begin{pmatrix}-1\\-1\\0\end{pmatrix},\quad
    &\vec s^{(3)}\vec v_4&=\frac{1}{\sqrt 3}\begin{pmatrix}1\\1\\0\end{pmatrix}
  \end{align*}
  it is sufficient to consider the scalar products for $l_1=1,l_2=2$ and
  $l_1=l_2=1$:
  \begin{align*}
    \left\langle d \mathcal E_{\vec u,T}^{\vec \alpha}\vec s^{(1)},d \mathcal
    E_{\vec u,T}^{\vec \alpha}\vec s^{(2)}\right\rangle
    &= \frac{3}{8}\sum_{\stackrel{j_1,j_2=1}{j_1\neq j_2}}^{4}-\frac 13\left\langle\vec s^{(1)}\vec v_{j_1},\vec v_{j_2}\right\rangle\left\langle\vec s^{(2)}\vec v_{j_2},\vec v_{j_1}\right\rangle=0,\\
    \begin{split}
      \left\langle d \mathcal E_{\vec u,T}^{\vec \alpha}\vec s^{(1)},d
        \mathcal E_{\vec u,T}^{\vec \alpha}\vec s^{(1)}\right\rangle
      &= \frac{3}{8}\sum_{\stackrel{j_1,j_2=1}{j_1\neq j_2}}^{4}-\frac 13\left\langle\vec s^{(1)}\vec v_{j_1},\vec v_{j_2}\right\rangle\left\langle\vec s^{(1)}\vec v_{j_2},\vec v_{j_1}\right\rangle\\
      &\quad+\frac{3}{8}\sum_{j_1=1}^{4}\frac{ 2}{6}+\frac{3}{4^2}\sum_{\stackrel{j_1,j_2=1}{j_1\neq j_2}}^{4}\frac 19\left\langle\vec s^{(1)}\vec v_{j_1},\vec s^{(1)}\vec v_{j_2}\right\rangle\end{split}\\
    &=\frac 38\cdot \frac 13\cdot \frac{64}{18}+\frac{3}{8}\cdot\frac{4\cdot 2}{6 }+\frac{3}{4^2}\cdot \frac 19\cdot\frac {-16}{6}=\frac 89.
  \end{align*}
  Hence, with $\beta_1=\frac{ 3}{2\sqrt 2}$ the proposed embedding is locally isometric.
\end{proof}

\changed{Finally, we consider icosahedral symmetry $Y$.
\begin{theorem}
  \label{theorem:isoY}
  Let $\vec u = \begin{pmatrix}0 \\ 1\\ \Phi \end{pmatrix}$, where $\Phi = \frac{1+\sqrt{5}}{2}$ is the golden ratio,
  $\vec \alpha = 10$ and $\vec \beta =\frac{75}{8\sqrt{95}} $. Then
  $\mathcal E_{\vec u,Y}^{\vec \alpha, \vec \beta}$ is a locally isometric embedding.
\end{theorem}
\begin{proof}
This proof is similar to the proof of the tetrahedral symmetry $T$.
The vectors $\vec v_i\in\left\{\vec R\vec u_1\right\}_{\vec R\in Y}$ are
  \begin{equation*}
    \begin{pmatrix}
      0\\ \pm1\\ \pm \Phi
    \end{pmatrix},\quad
    \begin{pmatrix}
      \pm 1\\\pm \Phi\\0
    \end{pmatrix},\quad
    \begin{pmatrix}
      \pm \Phi\\0\\ \pm 1
    \end{pmatrix},
  \end{equation*}
  and satisfy $|\langle\vec v_i,\vec v_j\rangle|=5^{-1/2}$ for $i\neq
  j$. Since $\alpha$ is even, we have
  $\otimes^\alpha(\vec x)=\otimes^\alpha(-\vec x)$ and do not have to consider
  $-\vec x$, if we use $\vec x$. Hence, we only need six vectors $\vec
  v_i$. Again, with Lemma~\ref{lem:vecpro} we can calculate
  $\d \mathcal E_{\vec u,Y}^{\alpha,\beta}$ and the scalar products of
  these. We omit these calculations here, as they are similar to the case for
  the tetrahedral symmetry $T$, but with higher-dimensional tensors.
\end{proof}}

\subsection{Global Inequalities}
\label{sec:glob-almost-isom}

Although the embeddings found in the previous section are \changed{locally} isometric they
obviously do not preserve the metric globally. In this section we are
interested in inequalities of the form
\begin{equation}\label{eq:ineq}
  c_{\min} \, d([\vec O_1]_S,[\vec O_2]_S)
  \le d(\mathcal E_S([\vec O_1]_S),\mathcal E_S([\vec O_2]_S))
  \le c_{\max}\,d([\vec O_1]_S,[\vec O_2]_S)
\end{equation}
that relate the Euclidean distance in $\R^{3^{\vec \alpha}}$ and the geodesic distance \changed{from equation~\eqref{eq:distSO3S}}.

The situation is easiest for $\mathcal S = C_{1}$, i.e., we just look at
$\SO$. In this case the Euclidean distance in the embedding is directly
related to the geodesic distance on the manifold via
\begin{equation*}
  d(\mathcal E_{C_1}(\vec R_1),\mathcal E_{C_1}(\vec R_2))
  =\changed{\sqrt{2}}\,\sqrt{1-\cos(d(\vec R_1,\vec R_2))}.
\end{equation*}
and we have $c_{\text{min}}= \changed{\frac{2}{\pi}}$ and $c_{\text{max}} = 1$.

For higher symmetries there is no such one to one relationship. In order to
illustrate the dependency between the geodesic distance on the manifold and
the Euclidean distance in the embedding for higher symmetries we have
visualized the regions of suitable combinations in Fig.~\ref{fig:picsold}
and \ref{fig:picsnew}. While Fig.~\ref{fig:picsold} illustrates the
embeddings from \cite{ArJuSc18}, Fig.~\ref{fig:picsnew} visualizes the
\changed{locally} isometric embeddings from Table~\ref{tab:embnew}.

\begin{figure}[p]%
\begin{subfigure}[c]{0.32\textwidth}
\includegraphics[width=\textwidth]{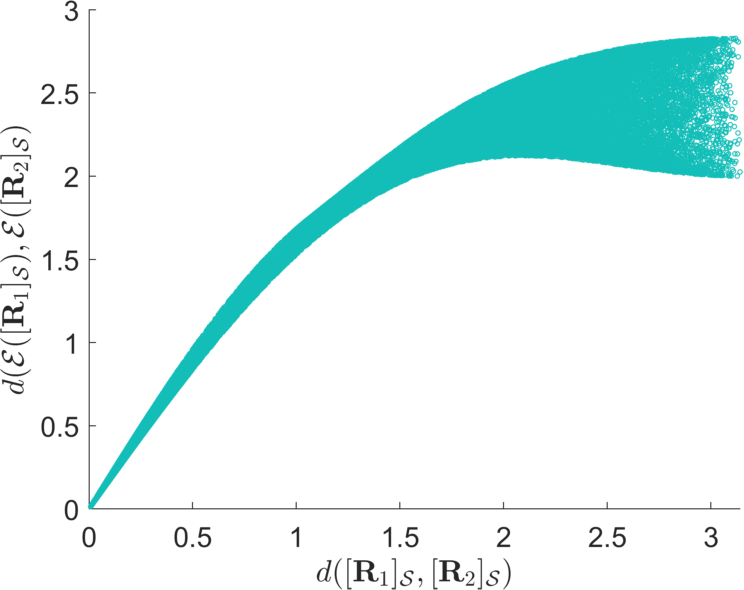}
\subcaption{Symmetry Group $C_2$}
\end{subfigure}
\begin{subfigure}[c]{0.32\textwidth}
\includegraphics[width=\textwidth]{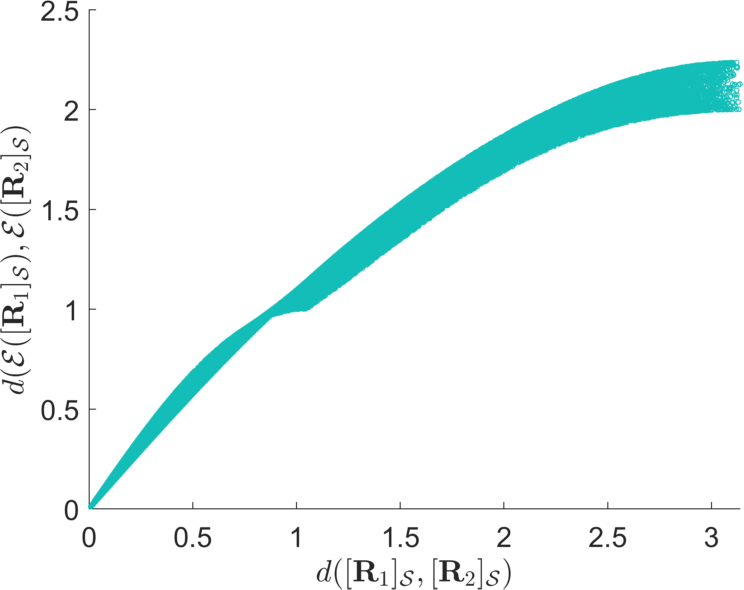}
\subcaption{Symmetry Group $C_3$}
\end{subfigure}
\begin{subfigure}[c]{0.33\textwidth}
\includegraphics[width=\textwidth]{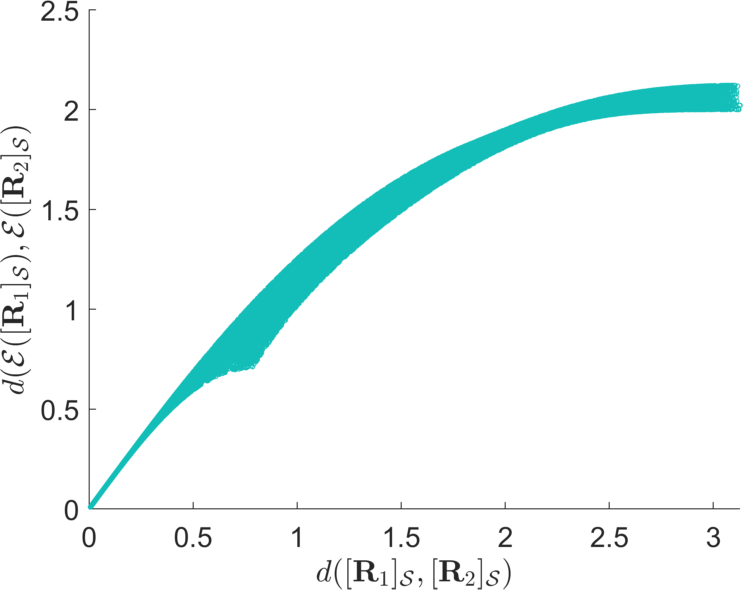}
\subcaption{Symmetry Group $C_4$}
\end{subfigure}
\begin{subfigure}[c]{0.32\textwidth}
\includegraphics[width=\textwidth]{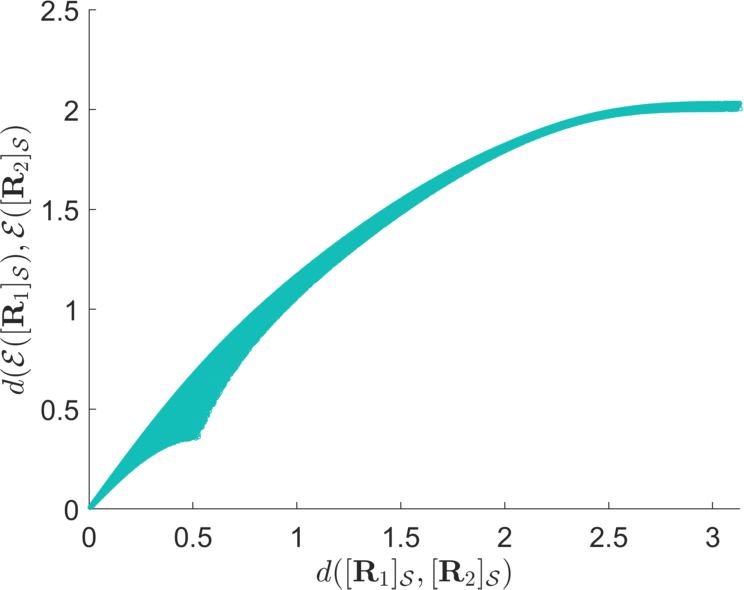}
\subcaption{Symmetry Group $C_6$}
\end{subfigure}
\begin{subfigure}[c]{0.32\textwidth}
\includegraphics[width=\textwidth]{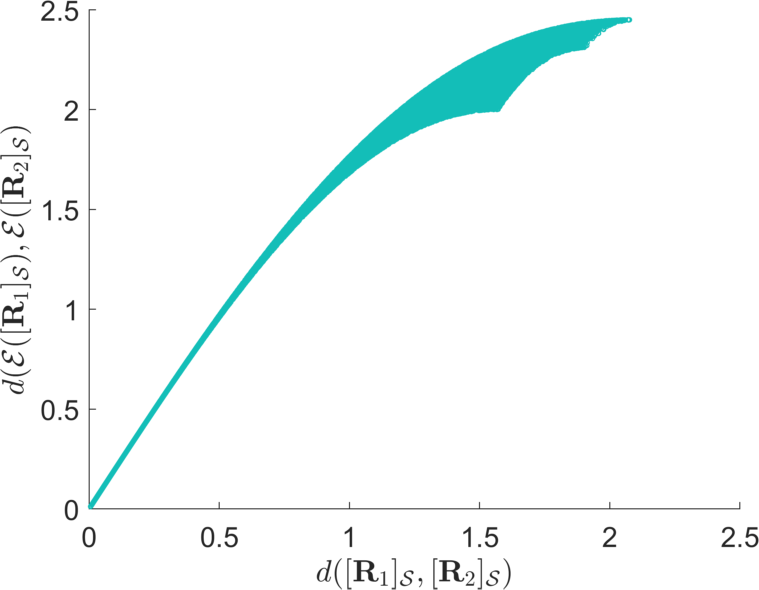}
\subcaption{Symmetry Group $D_2$}
\end{subfigure}
\begin{subfigure}[c]{0.32\textwidth}
\includegraphics[width=\textwidth]{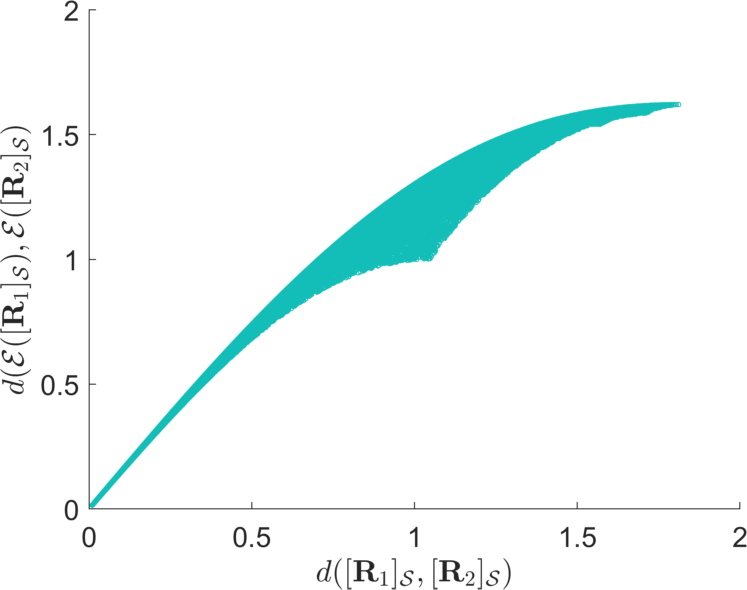}
\subcaption{Symmetry Group $D_3$}
\end{subfigure}
\begin{subfigure}[c]{0.32\textwidth}
\includegraphics[width=\textwidth]{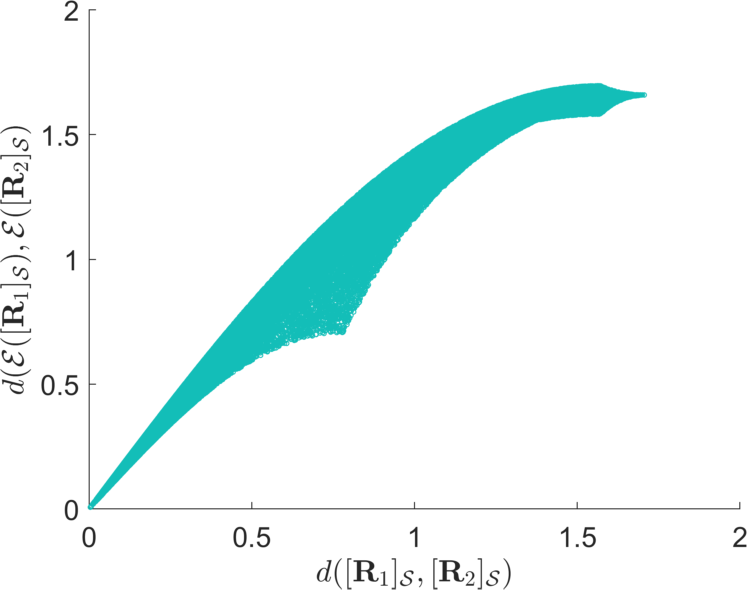}
\subcaption{Symmetry Group $D_4$}
\end{subfigure}
\begin{subfigure}[c]{0.32\textwidth}
\includegraphics[width=\textwidth]{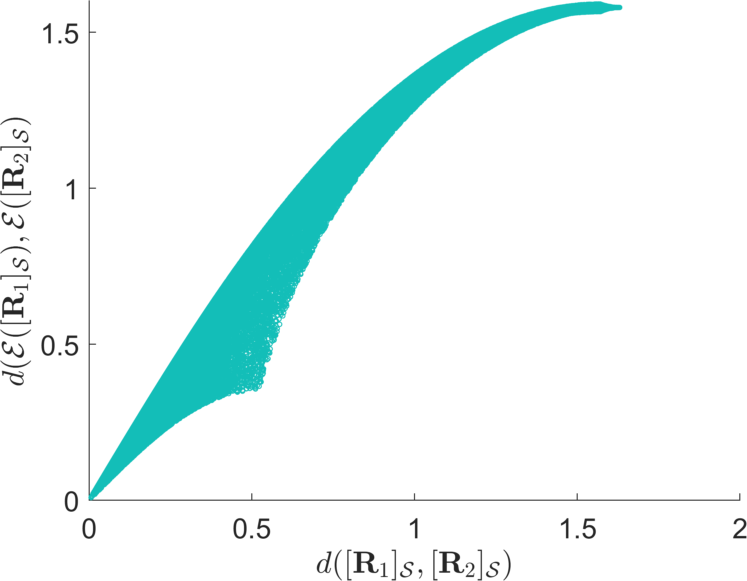}
\subcaption{Symmetry Group $D_6$}
\end{subfigure}
\caption{Relation between the geodesic distance on the manifold and the
  Euclidean distance in the embedding for the embeddings reported
  in~\cite{ArJuSc18}.}%
\label{fig:picsold}%
\end{figure}

\begin{figure}[p]%
\begin{subfigure}[c]{0.32\textwidth}
\includegraphics[width=\textwidth]{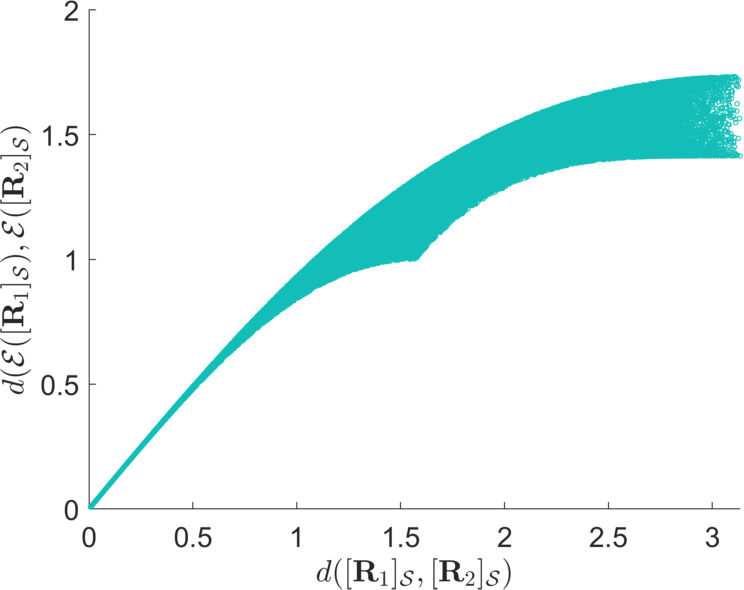}
\subcaption{Symmetry Group $C_2$}
\end{subfigure}
\begin{subfigure}[c]{0.32\textwidth}
\includegraphics[width=\textwidth]{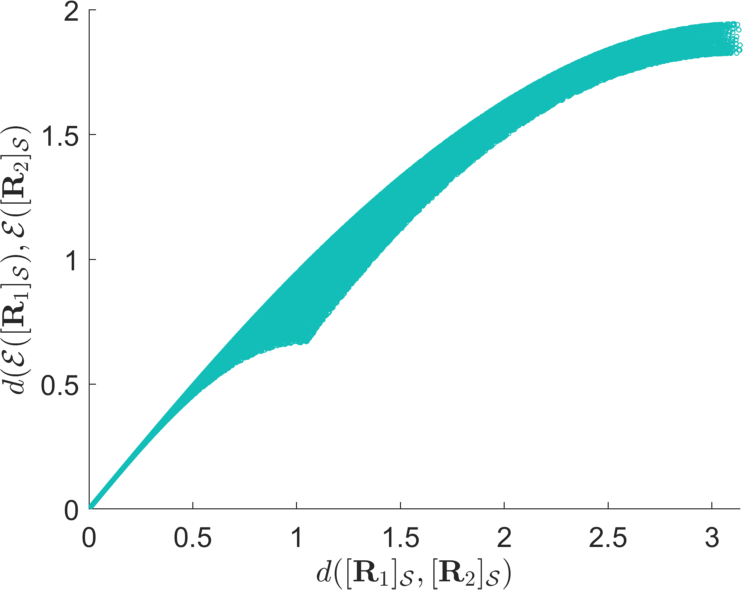}
\subcaption{Symmetry Group $C_3$}
\end{subfigure}
\begin{subfigure}[c]{0.33\textwidth}
\includegraphics[width=\textwidth]{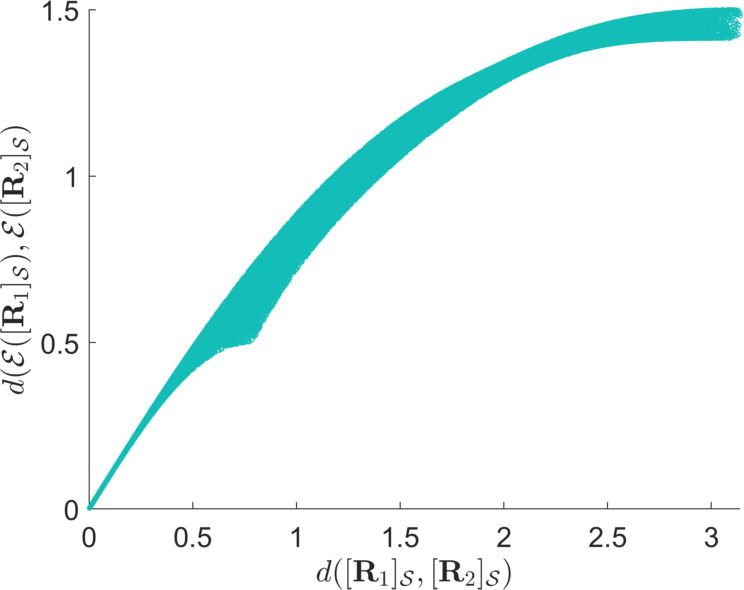}
\subcaption{Symmetry Group $C_4$}
\end{subfigure}
\begin{subfigure}[c]{0.32\textwidth}
\includegraphics[width=\textwidth]{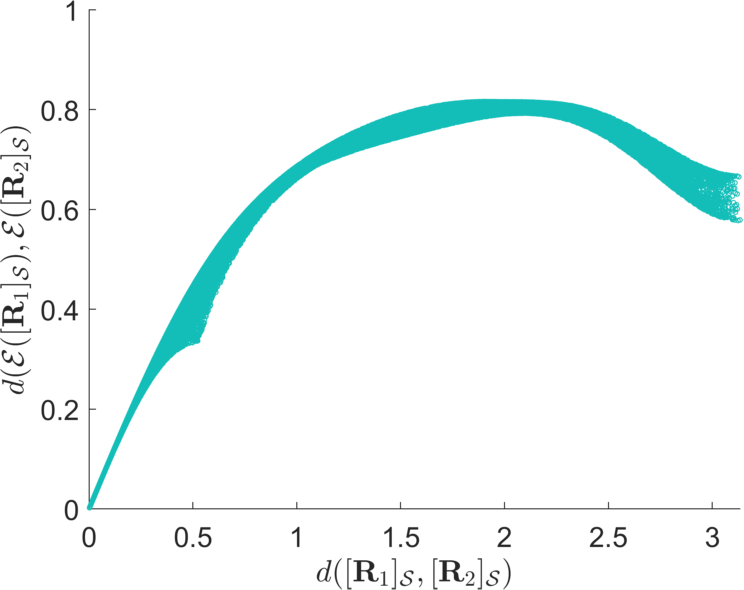}
\subcaption{Symmetry Group $C_6$}
\end{subfigure}
\begin{subfigure}[c]{0.32\textwidth}
\includegraphics[width=\textwidth]{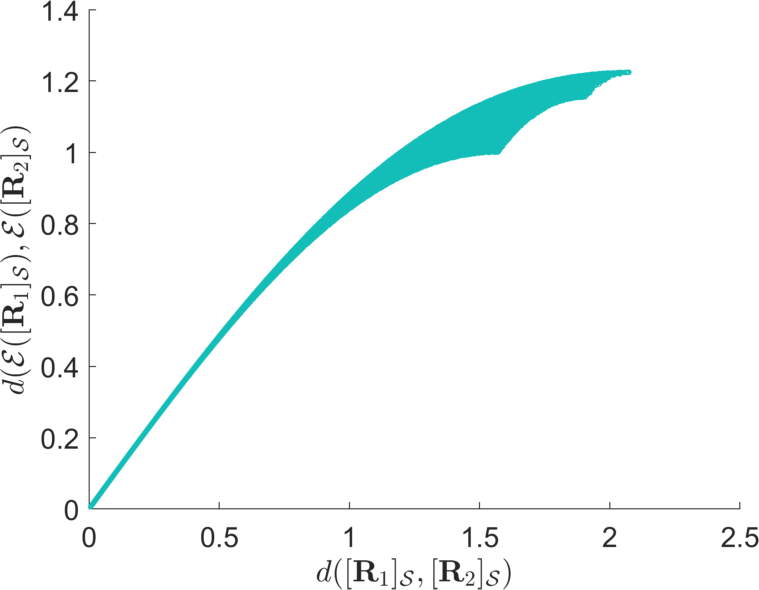}
\subcaption{Symmetry Group $D_2$}
\end{subfigure}
\begin{subfigure}[c]{0.32\textwidth}
\includegraphics[width=\textwidth]{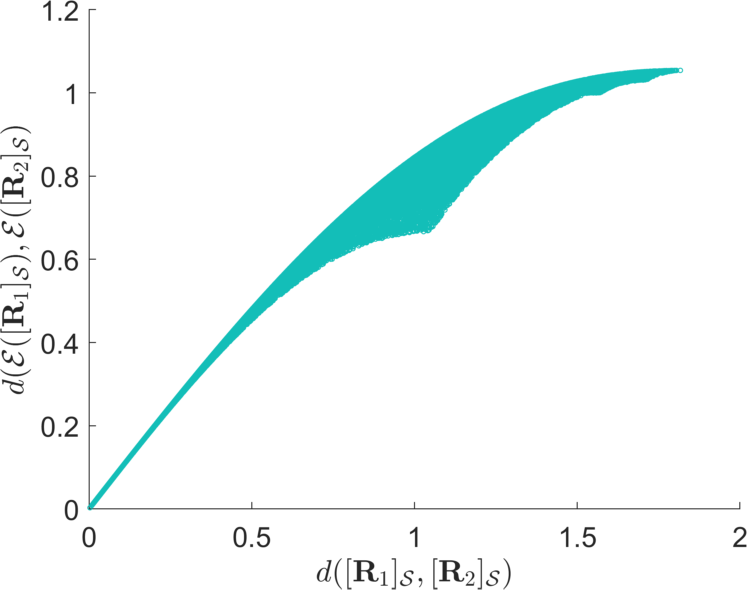}
\subcaption{Symmetry Group $D_3$}
\end{subfigure}
\begin{subfigure}[c]{0.32\textwidth}
\includegraphics[width=\textwidth]{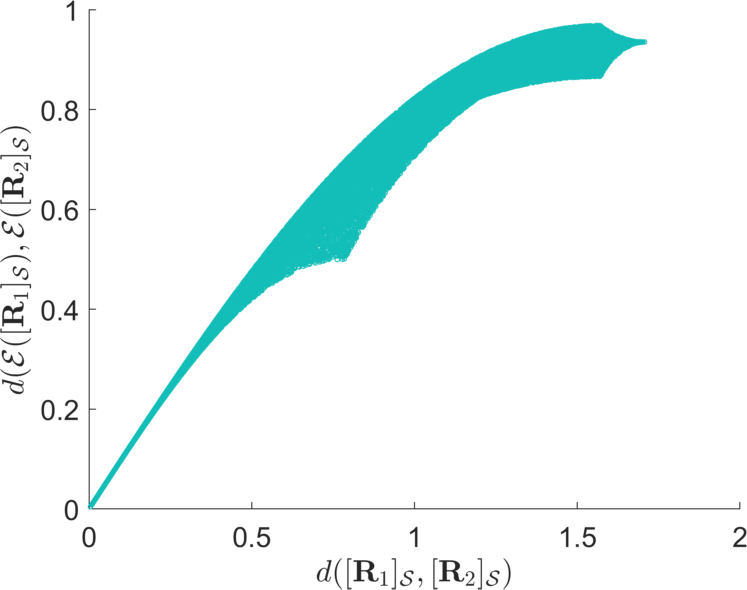}
\subcaption{Symmetry Group $D_4$}
\end{subfigure}
\begin{subfigure}[c]{0.32\textwidth}
\includegraphics[width=\textwidth]{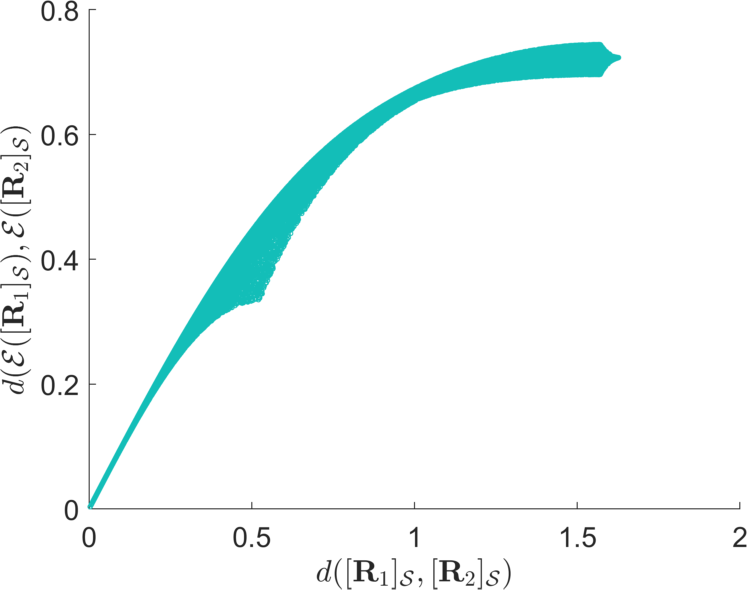}
\subcaption{Symmetry Group $D_6$}
\end{subfigure}
\begin{subfigure}[c]{0.32\textwidth}
\includegraphics[width=\textwidth]{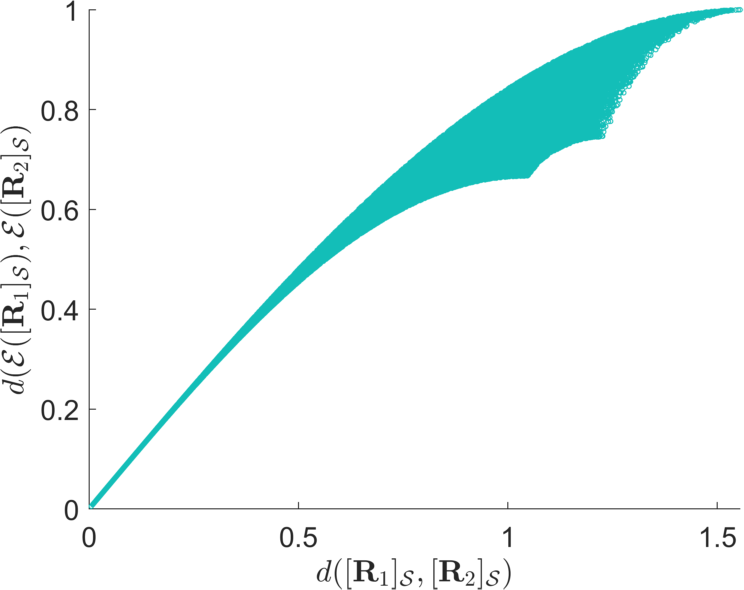}
\subcaption{Symmetry Group $T$}
\end{subfigure}
\begin{subfigure}[c]{0.32\textwidth}
\includegraphics[width=\textwidth]{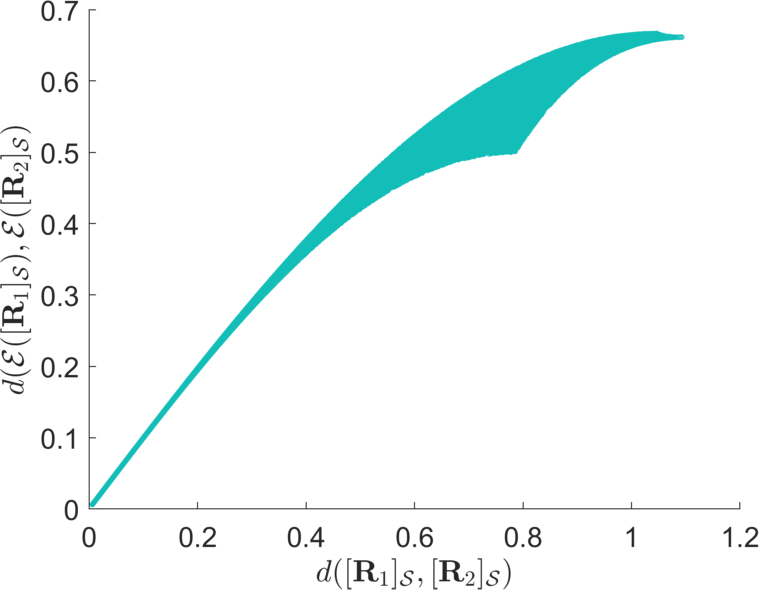}
\subcaption{Symmetry Group $O$}
\end{subfigure}
\caption{Relation between the geodesic distance on the manifold and the
  Euclidean distance in the embedding for the \changed{locally} isometric embeddings summarized
  in Table~\ref{tab:embnew}.}%
\label{fig:picsnew}%
\end{figure}

In Table~\ref{tab:const} the upper and lower bounds
 $c_{\text{min}}$ and $c_{\text{max}}$ are listed for \changed{locally} isometric embeddings
 from Table~\ref{tab:embnew} . We would like to stress that non-\changed{locally}-isometric
 embeddings might very well lead to better global bounds. Indeed,
 Table~\ref{tab:cmincmax} provides alternative coefficients for the embeddings
 $\mathcal E_{\vec u,\mathcal S}^{\vec \alpha, \vec \beta}$ which have better
 upper and lower bounds.



\begin{table}[ht]
  \caption{The constants in equation~\eqref{eq:ineq} for all crystallographic
    symmetry groups $\mathcal S$}
  \label{tab:const}
	\vspace{-0.2cm}
  \centering
  \begin{tabular}{cccc}
    \hline
    $\mathcal S$   & $c_{\min}$& $c_{\max}$&$c_{\max}/c_{\min}$\\
    \hline
    $C_2$ & $0.452$ 	& $1$ &	$2.21$\\
    $C_3$ & $0.583$	& $1$ &	$1.72$\\
    $C_4$	&	$0.452$	&	$1$	&	$2.21$\\
    $C_6$	&	$0.186$	&	$1$	&	$5.36$\\
    $D_2$ & $0.590$	&	$1$	& $1.70$\\
    $D_3$ & $0.581$	& $1$	& $1.72$\\
    $D_4$ &	$0.546$	&	$1$	& $1.83$\\
    $D_6$ &	$0.443$	&	$1$	& $2.26$\\
    $O$		&	$0.604$	&	$1$	& $1.66$\\
    $T$		&	$0.609$	&	$1$	& $1.64$\\
    \hline
  \end{tabular}

\end{table}

\begin{table}[ht]

\caption{
Factors for globally almost isometric embeddings for some symmetry
    groups $\mathcal S$}
		\label{tab:cmincmax}
\vspace{-0.2cm}

		\centering
  \begin{tabular}{ccc}
    \hline
     $\mathcal S$   & $\vec \beta$& $c_{\max}/c_{\min}$\\
    \hline
    $C_2$  	& $(1,0.5,0.5)$ &	$1.92$\\
    $C_3$ 	& $(1,0.67)$ 	&	$1.68$\\
		$C_4$		&	$(1,0.6)$		&	$1.91$\\
		$C_6$		&	$(1,0.93)$	& $2.15$\\
    $D_3$ 	& $(1,1.03)$	& $1.72$\\
		$D_4$ 	&	$(1,1.11)$	& $1.80$ \\
		$D_6$ 	&	$(1,1.65)$	& $1.95$\\
    \hline
  \end{tabular}

\end{table}

\section*{\changed{Acknowledgments}}
We thank the editor and the referees for their helpful comments and valuable suggestions. The second author acknowledges funding by Deutsche Forschungsgemeinschaft (DFG, German Research Foundation) - Project-ID 416228727- SFB 1410.

\appendix{}

\section{A binomial identity}
For the calculation of $\norm{M_\alpha}$ in Lemma~\ref{lem:skaMr} we need the following nice Lemma for binomial coefficients.

\begin{lemma}
  \label{lemma:sumBinom}
  Let $\alpha\in 2\N$ be an even integer. Then we have the equality
  \begin{equation*}
    \left(\alpha+1\right)\binom{\alpha}{\frac \alpha2}
    =\sum_{\substack{i_1,i_2,i_3=0\\i_1+i_2+i_3=\frac  \alpha2}}^{\frac  \alpha2}\binom{2i_1}{i_1}\binom{2i_2}{i_2}\binom{2i_3}{i_3}.
  \end{equation*}
\end{lemma}

\begin{proof}
  With the general definition of the binomial coefficient
  $\binom{n}{k}=\frac{n(n-1)\cdots(n-(k-1))}{k!}$ for $k>0$ we obtain
  \begin{equation*}
    \binom{2n}{n}=(-1)^n\cdot 4^n\cdot\binom{-\frac 12}{n}.
  \end{equation*}
  With this equation and the Chu-Vandermonde-identity it follows that
  \begin{align*}
    \sum_{\substack{i_1,i_2,i_3=0\\i_1+i_2+i_3=\frac \alpha2}}^{\frac \alpha2}\binom{2i_1}{i_1}\binom{2i_2}{i_2}\binom{2i_3}{i_3}
    &=\sum_{\substack{i_1,i_2,i_3=0\\i_1+i_2+i_3=\frac \alpha2}}^{\frac \alpha2}\left(-1\right)^{i_1+i_2+i_3}\cdot 4^{i_1+i_2+i_3}\binom{-\frac 12}{i_1}\,\binom{-\frac 12}{i_2}\,\binom{-\frac 12}{i_3}\\
    &=\left(-1\right)^{\frac \alpha2}\cdot 4^{\frac \alpha2}\sum_{\substack{i_1,i_2,i_3=0\\i_1+i_2+i_3=\frac \alpha2}}^{\frac \alpha2}\binom{-\frac 12}{i_1}\,\binom{-\frac 12}{i_2}\,\binom{-\frac 12}{i_3}\\
    &=\left(-1\right)^{\frac \alpha2}\cdot 4^{\frac \alpha2}\binom{-\frac 32}{\frac \alpha2}\\
    &=\left(-1\right)^{\frac \alpha2}\cdot 4^{\frac \alpha2}\,\left(\frac{-\frac 32\left(-\frac 32 -1\right)\cdots\left(-\frac 32 -(\frac \alpha2 -1)\right)}{\left(\frac \alpha2\right)!}\right)\\
    &=4^{\frac \alpha2}\,\left(\frac{\frac 32\left(\frac 32 +1\right)\cdots\left(\frac 32 +(\frac \alpha2 -1)\right)}{\left(\frac \alpha2\right)!}\right)\\
    &=2^{\frac \alpha2}\,\left(\frac{3\cdot 5\cdot 7\cdots\left( \alpha+1\right)}{\left(\frac \alpha2\right)!}\right)\\
    &=\left( \alpha+1\right)\frac{2^{\frac \alpha2}\left(\frac \alpha2\right)!\cdot3\cdot 5\cdot 7\cdots\left( \alpha-1\right)}{\left(\frac \alpha2\right)!^2}\\
    &=\left( \alpha+1\right)\binom{ \alpha}{\frac  \alpha2}.
  \end{align*}
\end{proof}

\section{\changed{Some trigonometrical sums}}

Here we calculate some trigonometric sums of the proofs in section \ref{sec:metric-properties}. For the proof of Lemma~\ref{lem:orthCk} we need
\begin{align}\label{eq:A1}
  \sum_{j_1,j_{2}=0}^{k-1} \cos^{k-1} \tfrac{2\pi(j_{1}-j_{2})}{k}
  \sin \tfrac{2\pi j_{1}}{k} \cos \tfrac{2\pi j_{2}}{k}
  &= \frac12 \sum_{j_1,j_{2}=0}^{k-1} \cos^{k-1} \tfrac{2\pi(j_{1}-j_{2})}{k}
    \left(\sin \tfrac{2\pi (j_{1}- j_{2})}{k} + \sin \tfrac{2\pi (j_{1} +
    j_{2})}{k} \right)\notag\\
  &= \frac12 \sum_{j_1,j_{2}=0}^{k-1} \cos^{k-1} \tfrac{2\pi j_{1}}{k}
    \left(\sin \tfrac{2\pi j_{1}}{k} + \sin \tfrac{2\pi j_{2}}{k} \right)
    = 0.
\end{align}
\changed{For the proof of Lemma~\ref{lem:normBl} we need the following calculations. Using
\begin{align*}
  \sum_{j=0}^{k-1} e^{\frac{2\pi ij n}{k}}
  =
  \begin{cases}
    k& n\in \Z \\
    0& else
  \end{cases}
\end{align*}
we compute
\begin{align*}
  \sum_{j=0}^{k-1}\cos^k\left(\frac{2\pi j}{k}\right)
  &=\frac{1}{2^k}\sum_{j=0}^{k-1}\left(e^{\frac{2\pi i j}{k}}+e^{\frac{-2\pi i j}{k}}\right)^k
    =\frac{1}{2^k}\sum_{j=0}^{k-1}\sum_{\ell=0}^k \binom{k}{\ell}e^{\frac{4\pi i j\ell}{k}}\\
  &=\frac{1}{2^k}\sum_{\ell=0}^k\binom{k}{\ell}\sum_{j=0}^{k-1} e^{\frac{2\pi i j(2\,\ell)}{k}}
    =\begin{cases}
      \frac{k}{2^{k-1}} &\text{if } k \text{ odd}\\
      \frac{1}{2^k}\left(\binom{k}{\frac k2}\cdot k+2k\right) &\text{if } k \text{ even}
    \end{cases},\\
  \sum_{j=0}^{k-1}\cos^{k-2}\left(\frac{2\pi j}{k}\right)
  &=\frac{1}{2^{k-2}}\sum_{j=0}^{k-1}\left(e^{\frac{2\pi i j}{k}}+e^{\frac{-2\pi i j}{k}}\right)^{k-2}
    =\frac{1}{2^{k-2}}\sum_{j=0}^{k-1}\sum_{\ell=0}^{k-2} \binom{k-2}{\ell}e^{\frac{2\pi i j(2\,\ell+2)}{k}}\\
  &=\frac{1}{2^{k-2}}\sum_{\ell=0}^{k-2}\binom{k}{\ell}\sum_{j=0}^{k-1} e^{\frac{2\pi i j(2\,\ell+2)}{k}}
    =\begin{cases}
      0 &\text{if } k \text{ odd}\\
      \frac{k}{2^{k-2}}\binom{k-2}{\frac k2 -1} &\text{if } k \text{ even}
    \end{cases}.
\end{align*}}
We use this for the following calculations.
\begin{align}\label{eq:A2}
 &-\frac{(k-1)}{k}\sum_{j_1=0}^{k-1}\sum_{j_2=0}^{k-1}\cos^{k-2}\left(\frac{2\pi(j_1-j_2)}{k}\right)\sin^2\left(\frac{2\pi(j_1-j_2)}{k}\right)
+\frac{1}{k}\sum_{j_1=0}^{k-1}\sum_{j_2=0}^{k-1}\cos^{k-1}\left(\frac{2\pi(j_1-j_2)}{k}\right)\cos\left(\frac{2\pi(j_1-j_2)}{k}\right)\notag \\
&\quad= -(k-1)\sum_{j=0}^{k-1}\cos^{k-2}\left(\frac{2\pi j}{k}\right)\sin^2\left(\frac{2\pi j}{k}\right)+\sum_{j=0}^{k-1}\cos^{k}\left(\frac{2\pi j}{k}\right)\notag\\
&\quad=-(k-1)\sum_{j=0}^{k-1}\cos^{k-2}\left(\frac{2\pi j}{k}\right)+k\,\sum_{j=0}^{k-1}\cos^{k}\left(\frac{2\pi j}{k}\right)\notag\\
&\quad=\begin{cases}
\frac {k^2}{2^{k-1}} &\text{if } k \text{ odd}\\
-\frac{k(k-1)}{2^{k-2}}\binom{k-2}{\frac k2-1}+\frac {k^2}{2^{k}}\left(\binom{k}{\frac k2}+2\right) &\text{if } k \text{ even}
\end{cases}.
\end{align}
Also for the proof of Lemma~\ref{lem:normBl} we calculate the following.
\begin{align}\label{eq:A3}
&\frac{1}{k}\sum_{j_1=0}^{k-1}\sum_{j_2=0}^{k-1}\cos^{k-1}\left(\frac{2\pi(j_1-j_2)}{k}\right)\cos\left(\frac{2\pi j_1}{k}\right)\cos\left(\frac{2\pi j_2}{k}\right)\notag\\
&\quad=\frac{1}{k\, 2^{k+1}}\sum_{j_1=0}^{k-1}\sum_{j_2=0}^{k-1}\left(e^{\frac{2\pi i (j_1-j_2)}{k}}+e^{-\frac{2\pi i (j_1-j_2)}{k}}\right)^{k-1}\left(e^{\frac{2\pi i j_1}{k}}+e^{-\frac{2\pi i j_1}{k}}\right)\left(e^{\frac{2\pi i j_2}{k}}+e^{-\frac{2\pi i j_2}{k}}\right)\notag\\
&\quad=\frac{1}{k\, 2^{k+1}}\sum_{j_1=0}^{k-1}\sum_{j_2=0}^{k-1}\left(e^{\frac{2\pi i (j_1+j_2)}{k}}+e^{-\frac{2\pi i (j_1+j_2)}{k}}+e^{\frac{2\pi i (j_1-j_2)}{k}}+e^{\frac{2\pi i (j_2-j_1)}{k}}\right)\sum_{l=0}^{k-1}\binom{k-1}{l}e^{\frac{2\pi i (j_1-j_2)(2l+1)}{k}}\notag\\
&\quad=\frac{1}{k\, 2^{k+1}}\sum_{l=0}^{k-1}\binom{k-1}{l}\sum_{j_1=0}^{k-1}\sum_{j_2=0}^{k-1}\Big(e^{\frac{2\pi i ((j_1-j_2)(2l+1)+j_1+j_2)}{k}}+e^{\frac{2\pi i ((j_1-j_2)(2l+1)-j_1-j_2)}{k}}+e^{\frac{2\pi i (j_1-j_2)(2l+2)}{k}}+e^{\frac{2\pi i (j_1-j_2)(2l)}{k}}  \Big)\notag\\
&\quad=\frac{1}{k\, 2^{k+1}}\sum_{l=0}^{k-1}\binom{k-1}{l}\sum_{j_1=0}^{k-1}\sum_{j_2=0}^{k-1}\Big(e^{\frac{2\pi i (j_1(2l+2)-2lj_2)}{k}}+e^{\frac{2\pi i (2lj_1-j_2(2l+2))}{k}}+e^{\frac{2\pi i (j_1-j_2)(2l+2)}{k}}+e^{\frac{2\pi i (j_1-j_2)(2l)}{k}}  \Big)\notag\\
&\quad=\frac{k}{k\, 2^{k+1}}\sum_{l=0}^{k-1}\binom{k-1}{l}\sum_{j=0}^{k-1} e^{\frac{2\pi i j(2l+2)}{k}}+e^{\frac{2\pi i j(2l)}{k}}\notag \\
&\quad=\begin{cases}
\frac{k}{2^{k}}&\text{if } k \text{ is odd}\\
\frac{k}{2^{k+1}}\left(2+\binom{k-1}{\frac k2}+\binom{k-1}{\frac k2-1}\right)&\text{if } k \text{ is even}
\end{cases}.
\end{align}

\bibliographystyle{myjmva.bst}
\bibliography{references/references}

\end{document}